\documentclass[journal,twocolumn, 10pt]{IEEEtran}

\usepackage{epsfig,makeidx,color,epstopdf}
\usepackage{amsmath,amssymb,amsthm,bbm}
\usepackage{cite,graphicx}
\usepackage{enumerate}
\usepackage{cite}
\usepackage{amsfonts}
\usepackage{mathtools, cuted}
\usepackage{lipsum}
\usepackage{indentfirst}
\usepackage{subfig}
\usepackage[font=footnotesize]{caption} 
\usepackage{csquotes}
\usepackage{float}
\usepackage{comment}

\makeatletter
\def\footnoterule{\kern-3\p@
  \hrule \@width 2in \kern 2.6\p@} 
\makeatother

\newtheorem{proposition}{Proposition}
\newtheorem{lemma}{Lemma}

\usepackage{array}
\newcolumntype{P}[1]{>{\centering\arraybackslash}p{#1}}
\newcolumntype{M}[1]{>{\centering\arraybackslash}m{#1}}

\DeclareMathAlphabet{\mathpzc}{OT1}{pzc}{m}{it}

\def\be{ \begin{equation} }
\def\ee{ \end{equation} }
\def\bea{ \begin{eqnarray} }
\def\eea{ \end{eqnarray} }

\def\b0{{\bf 0}}

\ifCLASSOPTIONonecolumn
\else
  
\fi

\makeatletter
\newcommand*\dashline{\rotatebox[origin=c]{90}{$\dabar@\dabar@\dabar@$}}
\makeatother
\makeatletter
\newcommand*{\rom}[1]{\expandafter\@slowromancap\romannumeral #1@}
\makeatother

\begin{document}
\title{RIS-Assisted Coverage Enhancement\\in Millimeter-Wave Cellular Networks}
\author{Mahyar Nemati, Jihong Park, and Jinho Choi
\thanks{Copyright (c) 2015 IEEE. Personal use of this material is permitted. However, permission to use this material for any other purposes must be obtained from the IEEE by sending a request to pubs-permissions@ieee.org. }
\thanks{M. Nemati, J. Park, and J. Choi are with
the School of Information Technology, 
Deakin University, 
Geelong, VIC 3220, Australia (e-mail: nematim@deakin.edu.au, jihong.park@deakin.edu.au, jinho.choi@deakin.edu.au)}
}
\date{today}
\maketitle

\begin{abstract}
The use of millimeter-wave (mmWave) bandwidth is one key enabler to achieve the high data rates in the fifth-generation (5G) cellular systems. However, mmWave signals suffer from significant path loss due to high directivity and sensitivity to blockages, limiting its adoption within small-scale deployments. To enhance the coverage of mmWave communication in 5G and beyond, it is promising to deploy a large number of reconfigurable intelligent surfaces (RISs) that passively reflect mmWave signals towards desired directions. With this motivation, in this work we study the coverage of an RIS-assisted large-scale mmWave cellular network using stochastic geometry, and derive the peak reflection power expression of an RIS and the downlink signal-to-interference ratio (SIR) coverage expression in closed forms. These analytic results clarify the effectiveness of deploying RISs in the mmWave SIR coverage enhancement, while unveiling the major role of the density ratio between active base stations (BSs) and passive RISs. Furthermore, the results show that deploying passive reflectors is as effective as equipping BSs with more active antennas in the mmWave coverage enhancement. Simulation results confirm the tightness of the closed form expressions, corroborating our major findings based on the derived expressions.


\end{abstract}
{\IEEEkeywords Millimeter-wave (mmWave), reconfigurable intelligent surface (RIS), coverage, signal-to-interference ratio (SIR), stochastic geometry.}
\section{Introduction}
Millimeter-wave (mmWave) cellular networks are widely studied for the emerging fifth generation ($5$G) of mobile communication networks and beyond. The Asia-Pacific and Americas regions are supposed to
give rise to the greatest share of the total contribution of
mmWave communications to the gross domestic product (GDP), \$$212$ billion and \$$190$ billion,
respectively, over the period of $2020$ to $2034$ \cite{Ros}; with a compound annual growth rate of $31$\% in the volume of mobile data
traffic \cite{Eric1-1}.
These significant 
growths imply that within the next decades, mmWave cellular networks will have significantly drawn attention to deliver much higher data-rate and capacity
compared to current levels due
to the availability of wider bandwidths \cite{RAP1,AccessSurvey, OUT1}.

As a primary distinctive technical feature, mmWave band suffers from a higher path loss than sub-6 GHz band. As a result, the mmWave communication range is limited. Nevertheless, when the frequency increases, the wavelength decreases which results in antenna aperture reduction. Thanks to a short wavelength ($1-10$ mm), it is feasible to pack multiple antenna elements into limited space at mmWave transceivers \cite{Cov&rate}. With large antenna arrays, e.g., multiple-input multiple-output (MIMO), mmWave cellular
    systems can execute beamforming to provide an array gain that compensates the frequency dependent path loss and overcomes additional noise power \cite{Cov&rate,RAP2, RAP3}. 
However, the mmWave communication range is still restricted due to the mmWave propagation characteristics, e.g., scattering, diffraction, and penetration loss \cite{RAP2,RAP3}. For instance, communications in mmWave frequencies highly suffer from penetration losses resulting in a blockage effect which mainly affects the line-of-sight (LoS) path and non-LoS (NLoS)
path loss characteristics \cite{Cov&rate}. 

Wireless transmission through multiple identified paths utilizing active MIMO relaying has been proposed as a potential solution that can reduce the blockage effect and increase the diversity \cite{YangRel1, KAVANPHD}. However, in \cite{RISNTO,RISBasar1,RISRui1}, it is shown that full-duplex MIMO relaying has a number of drawbacks such as signal processing complexity, noise enhancement, power consumption and self-interference cancellations at the relay stations along with their high implementation costs. To this end, it would be desirable to control the propagation environment in those frequencies with simple low-cost full-duplex passive reflectors like what has been recently proposed as in \textbf{reconfigure intelligent surfaces (RISs)} \cite{RISBasar1} to mitigate the aforementioned drawbacks.

An RIS is a software-defined metasurface containing a large number of passive reflectors and has given rise to the emerging \enquote{smart radio environments} concept \cite{RISBasar1}. The recent advent of RISs
in wireless communications enables network
operators to control the reflection
characteristics of the radio waves in an energy efficient way~\cite{park2020extreme}. The passive reflectors in an RIS are intelligently controlled by a main integrated circuit (IC) to adjust phase-shift of an impinging signal. 
In other words, RIS can turn the
wireless environment, which is highly probabilistic in nature,
into a controllable and partially deterministic phenomenon \cite{RISNTO,wcnc}. 
%


In the literature there is a significant effort to model mmWave cellular networks under different circumstances using stochastic geometry \cite{Cov&rate,ref2,ref3,ref4}. 
As early works, general stochastic geometry frameworks of mmWave cellular network were proposed in \cite{Cov&rate,ref2} to model the static objects and corresponding blockage probability using the concept
of random shape theory. Moreover, the impact of relay on a multi-hop medium access control protocol for 60 GHz frequency was investigated in \cite{ref3} when the LoS path is blocked. In \cite{ref4,KAVANPHD} comprehensive coverage performance analysis of relay-assisted mmWave cellular networks were investigated.
However, the aforementioned studies did not consider the spatial randomness of RISs deployments. In addition, the study of impact on RIS deployment in mmWave cellular networks is limited. 
In \cite{GongSurvey,GenFadRIS, liu2020, wu2020}, general comprehensive overviews characterizing the performance of RIS-assisted communications affecting the propagation environments were provided. 
Moreover, in \cite{noma1,noma2,noma3}, the impact of RIS deployments for non-orthogonal
multiple access (NOMA) networks were assessed.
In \cite{R1}, the effect of large-scale deployment of RISs on the performance of
cellular networks was studied by modeling the blockages using the line Boolean model.
In \cite{Mimo-1}, an RIS-assisted MIMO framework was proposed to randomly serve users by jointly passive
beamforming weight at the RISs and detection weight vectors at
the users.
In \cite{R2}, an analytical probability framework of successful reflection of RIS for a given transmission was provided using point processes, stochastic geometry, and
random spatial processes.
In \cite{R3}, authors proposed a distributed RIS-empowered
communication network architecture, where multiple source destination pairs communicate through multiple distributed RISs.
In \cite{AsymSINR}, an optimal linear precoder along with an RIS deployment in a single cell for multiple users is used to improve the coverage performance of the communications.
%


In this paper, we aim at studying a general tractable framework for the coverage performance of the RIS-assisted mmWave cellular networks with a major focus on RIS and BS densities. We use stochastic geometry as a powerful tool to study the average signal-to-interference-ratio (SIR) behavior over many randomly distributed BSs, RISs, and users in a 2-dimensional (2D) space. In our proposed model, BSs are equipped with a steerable antenna array and are able to send two beams towards a user equipment (UE). 
One beam is transmitted directly towards the UE, i.e., referred to as path \textbf{A}; and the other beam is sent towards the RIS and then reflected to the UE, i.e., referred to as path \textbf{B}. %
The main contributions of this paper are listed as follows.
\begin{itemize}
    \item We propose a general tractable RIS-assisted approach for SIR coverage performance in mmWave cellular networks for the first time where the message is sent by the BS towards the UE through two different paths. We use a diversity technique in which the system profits from the maximum received SIR at the UE through either path \textbf{A} or path \textbf{B}.
    \item Since the reflected power of passive RIS-reflectors is largely affected by the distance between the active BS and the RIS due to large-scale fading, we provide the probability distribution function (PDF) of this distance as an important quantity and discuss its dependency on the RIS and BS densities. 
    \item Discrete time delay values corresponding to quantized phase-shifts at each RIS-reflector is elaborated for passive beamforming at the RIS towards the UE. In addition, the peak reflection power at the RIS is assessed. It is shown that the average peak reflection power at the RIS decreases when the active BS density decreases. However, the reflected power reduction can be compensated by employing RISs of a large number of passive reflectors.
    \item  A closed-form approximation, referred to as \textbf{Approximation-\rom{1}}, along with a lower bound approximation, referred to as \textbf{Approximation-\rom{2}}, is derived for the SIR coverage probability of the signal received by the UE from path \textbf{B}, i.e., RIS-assisted path. 
    \item  Finally, we show that the RIS-assisted model provides a great deal of flexibility to obtain a desired SIR gain. Furthermore, we show that when the active BS density decreases, the co-channel interference caused by active interferer BSs decreases faster than the reflected power from the RIS. As a result, the decrease of active BS density improves the SIR coverage probability in our RIS-assisted model.
\end{itemize}

The rest of the paper is organized as follows. In Section \rom{2}, we present the system model of a baseline and RIS-assisted mmWave cellular networks. Then, the principles of RIS-assisted model are discussed in Section \rom{3}. Subsequently, the SIR coverage analysis of the RIS-assisted model is provided in Section \rom{4}. A comprehensive discussion on both of the baseline and RIS-assisted models is given in Section \rom{5}. Simulation results and comparisons are presented in Section \rom{6}. Finally, Section \rom{7} concludes the paper.

\section{System Model}
In this section, we first present a baseline downlink mmWave cellular network, followed by introducing the RIS-assisted downlink mmWave network. The baseline network model follows the standard frameworks for stochastic geometric mmWave system analysis \cite{REF,ParkSG}, but for the reader's convenience, we briefly describe the basics.

%
 We provide a set of common suppositions used for both baseline and RIS-assisted models as follows. In a network, BSs are randomly located in a 2D space according to a homogeneous Poisson point process (PPP), denoted by $\Phi_{BS}$, with an intensity of $\lambda_{BS}$. The UEs are distributed independently in the area, and each UE communicates with the nearest BS to enjoy the least mean propagation loss.
The probability density function (PDF) of the distance between the UE and the nearest BS, denoted by $r_0$, is obtained from the void probability in Poisson process of $\mathbb{R}^2$ as follows \cite{PPP}:
\begin{equation}
    f_{r_0}(r_0)=2\pi \lambda_{BS} r_0 e^{-\lambda_{BS}\pi r_{0}^2}.
    \label{pdfR0}
\end{equation}
For the mathematical amenability, blockage effects are omitted in this study. Since reflections enable to overcome blockages as observed in \cite{R1}, we envisage that RISs will further enhance the network coverage under blockages, and investigating this is deferred to our future work.
%
%
%
 We assume that all the BSs are equipped with $N$ isotropic active elements for beamforming towards the targets while the UEs are equipped with single omnidirectional antenna. 
 Furthermore, we consider the following assumption for the small-scale fading channel gain.
 
{\it Assumption 1 (Small-Scale Channel Gain):} The small-scale fading gain is assumed to follow an exponential distribution with mean of $1/\mu$. In the past, this assumption was common in stochastic geometric coverage analysis for mathematical tractability~\cite{Cov&rate,parkSG1,ChanelSG2,parkSG2}. Recent works~\cite{ParkSG,limani2020mmwave} revisited this exponential fading assumption, and rediscovered its feasibility even under realistic large-scale mmWave systems, by simply tuning $\mu$ according to the mean channel characteristics and antenna patterns.

\subsection{Baseline mmWave Cellular Network}
\begin{figure}[t]
                \centering
                    \includegraphics[width=6.35cm, height=3.8cm]{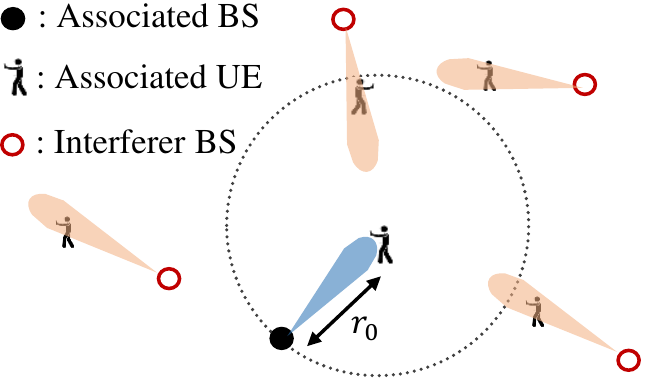}
                    \caption{Depiction of the baseline mmWave cellular network.}
                \label{fig:Conven-1}
\end{figure}
Fig. \ref{fig:Conven-1} depicts the baseline mmWave cellular network where each BS creates one beam to send the downlink signal towards the desired UE. The transmit signal power at the active BSs is assumed to be constant and is denoted by $P_s$.
 Compared to the sub-6 GHz cellular networks in \cite{REF}, in the mmWave cellular networks, beamforming is used to converge the signal power in a specific direction towards a desired UE due to the mmWave propagation characteristics \cite{parkBW}.
  The existing coverage analysis in \cite{REF} evaluated the conventional cellular networks for sub-6 GHz frequency bands. We exploit their analysis and modify it to introduce our baseline mmWave cellular network while taking beamforming into account. In general, as shown in Fig. \ref{fig:Conven-1}, there are two types of signal power sources. One type is the desired signal power received by the UE from the associated BS as the desired source, and the other type is the interference signal power received by the UE from the interferer BSs. 
 Based on the uniform planar square array (UPA) in 2D space \cite{parkBW}, the active BSs with $N$ isotropic elements are able to create single beam with beamwidth:
 \begin{equation}
     \psi_\textbf{o}=\frac{2\pi}{\sqrt{N}}.
     \label{BW1}
 \end{equation}
 Let each interferer BS transmit with its main-lobe pointed at a random direction. Nevertheless, as an additional gain, it also reduces co-channel interference because the signal from any NLOS interferer is highly attenuated \cite{kavan}. Intuitively, it affects the density of interferer BSs as a modified homogeneous PPP, denoted by $\Phi_I$, with an intensity of $\lambda_I$ as follows:
  \begin{equation}
      \lambda_I=\frac{\lambda_{BS}}{\sqrt{N}}.
      \label{lam_I}
  \end{equation}
 In other words, only a subset of interferer BSs in which their beamforming direction covers the desired UE are considered to be effective interferer BSs for the UE. 
 
  Let $\Gamma_\textbf{o}$ denote the SIR at an independent UE in the baseline mmWave cellular network. In this paper, we omit the noise power in SINR and evaluate the SIR-based performance for simplicity. Consequently, the SIR can be obtained as
\begin{equation}
    {\Gamma}_\textbf{o}=\frac{P_s g_0 r_0^{-\alpha}}{\sum\limits_{\substack{BS_i \in \Phi_{I},\\ i\neq 0}} P_s g_i r_i^{-\alpha}}=\frac{g_0 r_0^{-\alpha}}{\sum\limits_{\substack{BS_i \in \Phi_{I},\\ i\neq 0}}  g_i r_i^{-\alpha}},
    \label{S-A}
\end{equation}
where $g_i$ and $r_i$ are small-scale channel gain and the distance between the $i$th BS, denoted by $BS_i$, and the UE, respectively. Here, $i=0$ indicates the nearest BS which is the associated BS and $\alpha$ is the path-loss exponent of large-scale fading. From \textit{Assumption 1}, we have $g_i\sim \exp (\mu)$ for all $i$. 

 For the SIR coverage probability which is the probability that the received SIR is larger than a threshold, let $T$ denote the threshold. Then, from \eqref{S-A}, the SIR coverage probability is given by 
     \begin{align}
    \Pr\left[\Gamma_\textbf{o}>T\right]=\mathbb{E}\left\{\Pr\left(\frac{ g_0 r_0^{-\alpha} }{\sum\limits_{BS_i \in \Phi_{I}}  g_i r_i^{-\alpha}}>T\right)\right\}.
    \label{eq213}
\end{align}
In other words, it is equivalent with the complementary cumulative distribution function (CCDF) of SIR. Eventually, after some analysis given in Appendix \ref{AP0}, 
we have
\begin{align}
    \Pr\left[\Gamma_\textbf{o}>T\right]&=
         \frac{1}{1+\frac{1}{\sqrt{N}}T^{\frac{2}{\alpha}}
         \int_{T^{-\frac{2}{\alpha}}}^\infty \frac{1}{1+u^{\frac{\alpha}{2}}}du}
         \label{q2}
\end{align}
It is noteworthy that the final coverage probability expression is independent of the BSs' transmit power and density. It only depends on the beamwidth of the beams (i.e., $N$), $T$, and $\alpha$.
 
\begin{figure}[t]
                \centering
                    \includegraphics[width=6.35cm, height=3.8cm]{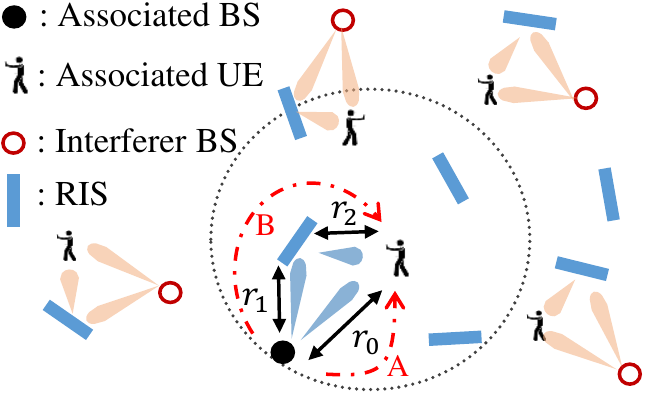}
                    \caption{Depiction of the RIS-Assisted Mm-Wave Cellular Network.}
                \label{system}
\end{figure}
 \subsection{RIS-Assisted mmWave Cellular Network}
 Suppose that there are buildings equipped with RISs in a mmWave cellular network which are distributed based on a homogeneous PPP, denoted by $\Phi_{RIS}$, with an intensity of $\lambda_{RIS}$. Each RIS has two parts: 1) passive part containing $M$ passive reflectors, and 2) a simple active part acting as a phase-shift controller.
 The UE communicates with the nearest RIS along with the nearest associated BS. 
Fig. \ref{system} shows the associated BS and the UE in coexistence of RISs and other interferer BSs. 
The distances between the associated BS and the associated RIS and the RIS and the UE are denoted by $r_{1}$ and $r_{2}$, respectively.
The PDF of $r_2$ can be obtained from the void probability in Poisson process of $\mathbb{R}^2$ as 
\begin{equation}
    f_{r_{2}}(r_{2})=2\pi \lambda_{RIS} r_{2} e^{-\lambda_{RIS}\,\pi r_{2}^2}.
    \label{pdfr2}
\end{equation}
Different from the baseline model, here when there is an RIS closer to the UE than the BS, i.e., $r_2<r_0$, the associated BS divides its single beam into two similar beams. The first beam is transmitted directly towards the desired UE, and the second beam targets the nearest RIS to the UE as shown in Fig. \ref{system}. As a result, the beamwidth of each of these two beams changes, from \eqref{BW1} in the baseline model, into 
\begin{equation}
    \psi_\textbf{s}=\frac{2\sqrt{2}\pi}{\sqrt{N}},
    \label{BW2}
\end{equation}
in 2D space and the transmit power of each beam at the active BSs becomes $\frac{P_s}{2}$. Nevertheless, this happens only when $r_2<r_0$. With the analysis given in Appendix \ref{AP00}, the probability of having an RIS within the distance between the associated BS and the UE becomes
\begin{equation}
f_{r_2}\left(r_2|r_2<r_0\right)=2\pi(\lambda_{RIS}+\lambda_{BS}) r_2 e^{-\pi(\lambda_{RIS}+\lambda_{BS})r_2^2}.
    \label{lem1}
\end{equation}
Throughout the paper, we consider $\lambda_{RIS}\gg\lambda_{BS}$, as shown in Fig. \ref{mesh}. Since RISs are passive, they are easier and cheaper to be implemented than the active BSs. Therefore, the expression in \eqref{lem1} approximately becomes equivalent to \eqref{pdfr2}, i.e.,
\begin{figure}[t]
\centering
\subfloat[Baseline model]{
 \includegraphics[width=4.4cm, height=3.8cm]{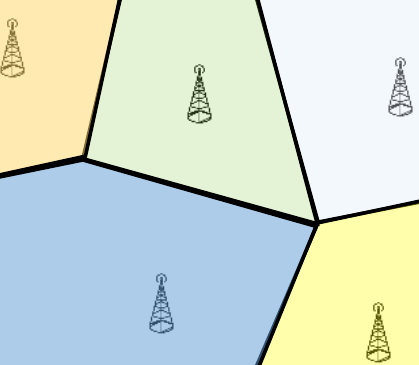}%
 }
 \subfloat[RIS-assisted model.]{
  \includegraphics[width=4.4cm, height=3.8cm]{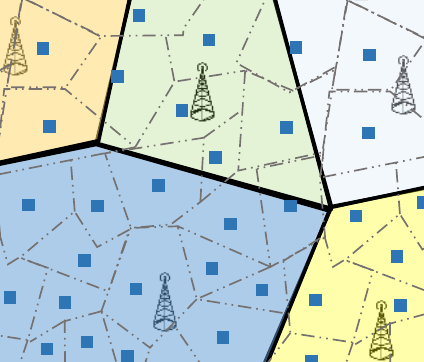}%
  }
 \captionsetup{width=1\linewidth}
 \caption{ Comparison between the baseline and RIS-assisted models when $\lambda_{RIS}\gg\lambda_{BS}$. Squares indicate RISs. In RIS-assisted model, each RIS has its own coverage area since the UE communicates with the nearest RIS.} 
 \label{mesh}
\end{figure}
\begin{equation}
    \lambda_{RIS}\gg\lambda_{BS} \rightarrow f_{r_2}\left(r_2|r_2<r_0\right)\approx f_{r_2}\left(r_2\right).
    \label{con1}
\end{equation}

  Thus, in general, there are two different paths from the BS towards the UE as shown in Fig. \ref{system}. One path which is directly from the BS to the UE, i.e., referred to as path \textbf{A}; and the other path which goes through the RIS and then reflected towards the UE, i.e., referred to as path \textbf{B}.
Let assume each RIS serves one UE at a time\footnote{One RIS can serve multiple UEs assuming the size of the RIS is determined based on the density of UEs around it. Hence, in high density UE areas, a larger RIS is needed where a portion of its reflectors can be dedicated to each UE. However, this can be considered as another resource allocation problem which its further assessment is out of the scope of this paper. Therefore, without loss of generality, we assume that each RIS serves one UE at a time.}.
The phase-shift controller at the RIS can adjust the phase-shifts and generate a new beam towards the UE (as described in section \ref{SBF}). In fact, the RIS acts as a passive beam-former by adjusting the phase-shifts at the passive reflectors and converge a beam in a specific direction.
 Thus, we define two states for RISs with respect to their reflection directions as shown in Fig. \ref{empty}. 
 \begin{itemize}
    \item Engaged: The nearest RIS to the UE which is engaged for the communication assistant in path \textbf{B}. 
    \item Idle: All other RISs which are not engaged in any communication are considered idle RISs. This is the default state when the phase-controller in an idle RIS adjusts the phase-shifts at the reflectors somehow to generate a beam towards an empty space, e.g., sky, to avoid interference with the UEs. Intuitively, it is evident that the idle RISs do not contribute to an interference.
\end{itemize}
 %
\begin{figure}[t]
                \centering
                    \includegraphics[width=7cm, height=4.5cm]{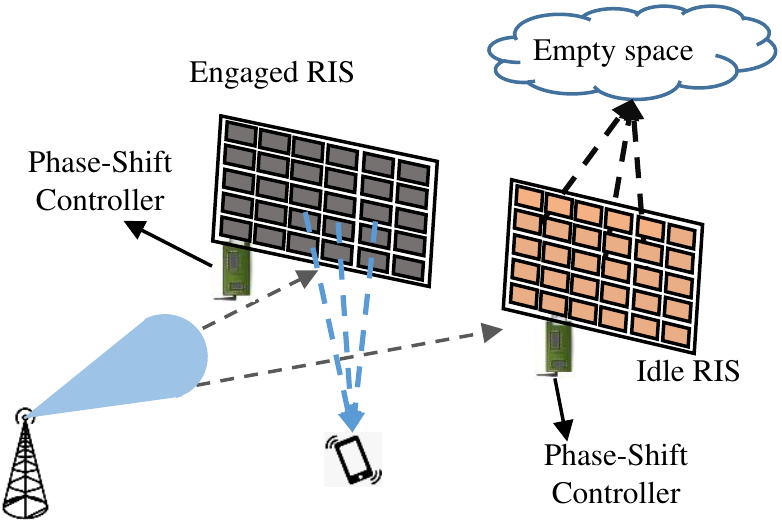}
                    \caption{The idle RIS reflects the signals toward the empty spaces, e.g. Sky.
}
                \label{empty}
\end{figure}

 In a nutshell, the anatomy of the communication initiation is briefly explained as follows.
\begin{itemize}
    \item The BS broadcasts the UE identification number as a paging message.
    \item Mobile receives the paging message and identifies itself along with its location to the BS and nearby RISs and a successful handshake between the BS and the UE takes place\footnote{It is noteworthy that angle of arrival of the UE's response can be simply obtained by both the BS and RIS with passive localization methods \cite{puo}.}. 
    \item Subsequently, the BS which knows the location of the UE and fixed RISs around it, informs the UE and its nearest RIS which forward channel the UE has been assigned. 
    \item The associated RIS becomes an engaged RIS and its phase-controller adjusts the phase-shifts at the reflectors.
\end{itemize}
\textit{Assumption 2 : } 
An engaged RIS is able to create a beam with highly narrow beamwidth toward the UE since it contains a large number of reflectors. Moreover, the reflection steering angle of each RIS is only limited to $[0, \quad \pi]$.
 Thus, it may receive signal only from a half of the interferer BSs\footnote{Besides, since the RISs are passive and their reflected power significantly being affected by large-scale path-loss, they can only cause an interference for nearby UEs. 
 }.
 Intuitively, it significantly reduces the probability that the engaged RIS contributes to an effective interference. Therefore, throughout the paper, we neglect the interference that may be caused by engaged RIS reflections\footnote{Recently, authors in \cite{RuiRISNet} showed that the interference marginally increases by RISs deployments. However, in order to maximize the mathematical tractability, we ignore this limited interference.}.

In the following section, we explain the principles of the RIS-assisted mmWave cellular networks in more details.
\section{Principles of RIS-Assisted MmWave Cellular Network}
In this section, first, we explain the phase-shift adjustment at each passive RIS-reflector and how passive beamforming is done by RISs. Second, a distribution function for the distance between the associated BS and the engaged RIS is provided. Finally, the peak reflection power at the RIS is assessed.

\subsection{Phase-Shift Adjustment at RIS and Passive Beamforming}\label{SBF}
               \begin{figure}[t]
                \centering
                    \includegraphics[width=6cm, height=3.5cm]{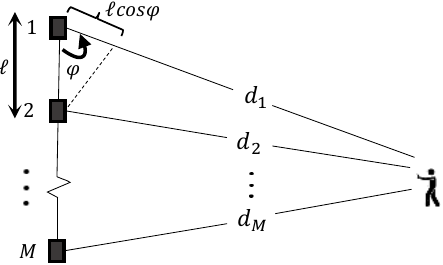}
                    \caption{Depiction of $M$-passive-elements linear phased-array at the RIS.}
                \label{fig:BF}
    \end{figure}
Seeing that the RIS includes $M$ passive reflective elements, it does not generate any transmit power by its own. 
The phase controller of the RIS determines the phase-shift for each reflected signal to create a new beam towards the UE. 
Let assume that it is only able to adjust discrete phase-shifts for the impinging signal at each passive reflector due to the implementation constraints. 
In particular, in order to show the effect of $M$ passive reflectors of the RIS on the transmit power, and also the phase-shift at each reflector, let simplify the generic 2D phased-array by applying it into a 1D linear phased-array \cite{BEAMBOOK} as shown in Fig. \ref{fig:BF}; 
where the distance between the passive reflectors is denoted by $\ell$ and distance of the $m$th passive reflector to the UE is denoted by $d_m$, where $m=1,\cdots , M$.
    Let $s(t+\delta_m)$ denote the message impinging the $m$th passive reflector at the RIS where $\delta_m$ stands for the phase difference of the impinging message at the $m$th reflector. Then, the reflected signal towards the UE which is a superposition due to $M$ passive reflectors is given by
    \begin{align}
        x(t)&=\sum_{m=1}^{M} s\left(t+\delta_m-\frac{d_m}{c}-\tau_m\right)\nonumber\\
        &=\sum_{m=1}^{M} s\left(t+\delta_m-\frac{d_1}{c}+\frac{m\ell\cos \varphi}{c}-\tau_m\right),
        \label{trx}
    \end{align}
    where $c$, $\varphi$ and $\tau_m$ correspond to the wave-speed, angle of the $1$th element reflection toward the UE (as shown in Fig. \ref{fig:BF}) and time delay of the $m$th reflector. Here, the time delay, $\tau_m$, is associated with the phase-shift adjustment at the $m$th reflector. In other words, $\tau_m$ is given by
    \begin{equation}
        \tau_m=\frac{m\ell\cos \varphi}{c}+\delta_m.
    \end{equation}
    Therefore, the reflected signal in \eqref{trx} becomes
        \begin{align}
        x(t)=\sum_{m=1}^{M} s\left(t-\frac{d_1}{c}\right)=Ms\left(t-\frac{d_1}{c}\right).
        \label{trx1}
    \end{align}
    
    Consequently, the RIS executes a new passive beamforming towards the UE, i.e., in an angle of $\varphi$, where the maximum power of the beam, i.e., peak effective radiated power, scales up by $M^2$.
       In addition, in order to put the phase-shift adjustment at each reflector into action, we may need to consider implementation constraints since $\frac{m\ell\cos \varphi}{c}+\delta_m$ might be continues values. However, we consider discrete values for $\tau_m$ that there will be a minimum phase-shift adjustment. We quantize the continues phase-shift amplitudes by discretizing them into the implementable values.

 \subsection{Distance Between the RIS and the Associated BS}
\begin{figure}[t]
                \centering
                    \includegraphics[width=4.6cm, height=4cm]{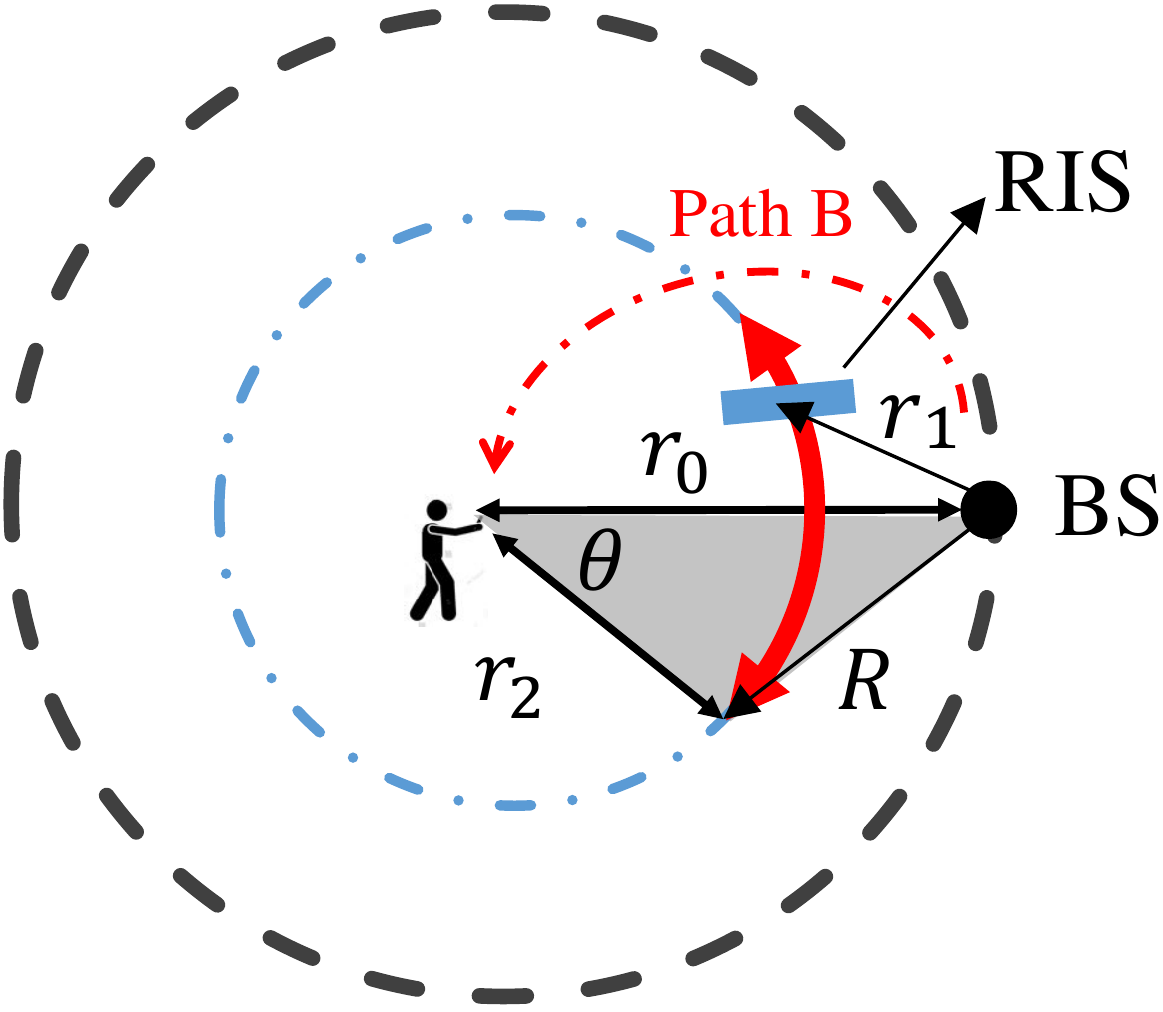}
                    \caption{The distance between the associated BS and the associated RIS is limited to $|r_0-r_2|<r_1<r_0+r_2$. The red captured annulus indicates where the potential RIS can be located when $r_1\leqslant R$ for a given $r_0$ and $r_2$.
}
                \label{fig:r_1}
\end{figure}
 An important quantity is the distance between the engaged RIS and the associated BS, denoted by $r_1$, which affects the transmit power at the RIS considering the large-scale fading. Although the BSs and RISs are both distributed based on two independent homogeneous PPPs, i.e., $\Phi_{BS}$ and $\Phi_{RIS}$, $r_1$ depends on both $r_0$ and $r_2$ as shown in Fig. \ref{fig:r_1}. Therefore, the distribution of $r_{1}$ for given $r_0$ and $r_2$ equals the captured annulus of the circumference of the circle with a radius of $r_{2}$ and origin of the UE as shown in Fig. \ref{fig:r_1}. Hence, the CDF of $r_{1}$ becomes the probability of the engaged RIS located on annulus of $2\theta r_{2}$ over the whole possible area of the circumference of $2\pi r_2$ as follows:
\begin{equation}
    \Pr\left(r_{1}<R|r_0,r_2\right)=F_{r_{1}}\left(R|r_0,r_2\right)=\frac{2\theta r_{2}}{2\pi r_{2}}=\frac{\theta}{\pi}.
    \label{er12}
\end{equation}
From the  law of Cosines in the shaded triangle\footnote{$a^2=b^2+c^2-2bc \cos (\hat{A})$}, we have
\begin{equation}
    \theta =\cos^{-1}\left(\frac{r_0^2+r_{2}^2-R^2}{2r_0r_{2}}\right).
    \label{tr34}
\end{equation}
 By substituting \eqref{tr34} in \eqref{er12} we have
\begin{align}
    &F_{r_{1}}\left(R|r_0,r_2\right)=\frac{\cos^{-1}\left(\frac{r_0^2+r_{2}^2-R^2}{2r_0r_{2}}\right)}{\pi}
\end{align}
Then, the PDF of $r_{1}$ for given $r_0$ and $r_2$ can be obtained as
\begin{align}
    f_{r_{1}}(r_{1}|r_{0},r_2)=\frac{d F_{r_{1}}(r_{1})}{d r_{1}}=\frac{r_{1}}{\pi r_0 r_{2}\sqrt{1-\left(\frac{r_0^2+r_{2}^2-r_{1}^2}{2r_0r_{2}}\right)^2}}.\label{PDFr1}
\end{align}

It is noteworthy that the $r_1$ can vary from $|r_0-r_2|$ to $r_0+r_2$. Eventually, from \eqref{pdfR0} and \eqref{pdfr2}, the PDF of the $r_1$ is given by 
\begin{align}
 &f_{r_{1}}(r_{1}) 
    = \int_{r_0=0}^\infty \int_{r_{2}=0}^{\infty} 
    f_{r_{1}}(r_{1}|r_{2},r_0) f(r_{2}) f(r_{0}) \,\,
    \, dr_{2}\,\, dr_0,\nonumber\\
    &\text{s.t.   } |r_0-r_{2}| \leqslant r_1 \leqslant r_0+r_{2}. 
    \label{pdr10}
\end{align}

\begin{figure}[t]
                \centering
                    \includegraphics[width=7cm, height=6cm]{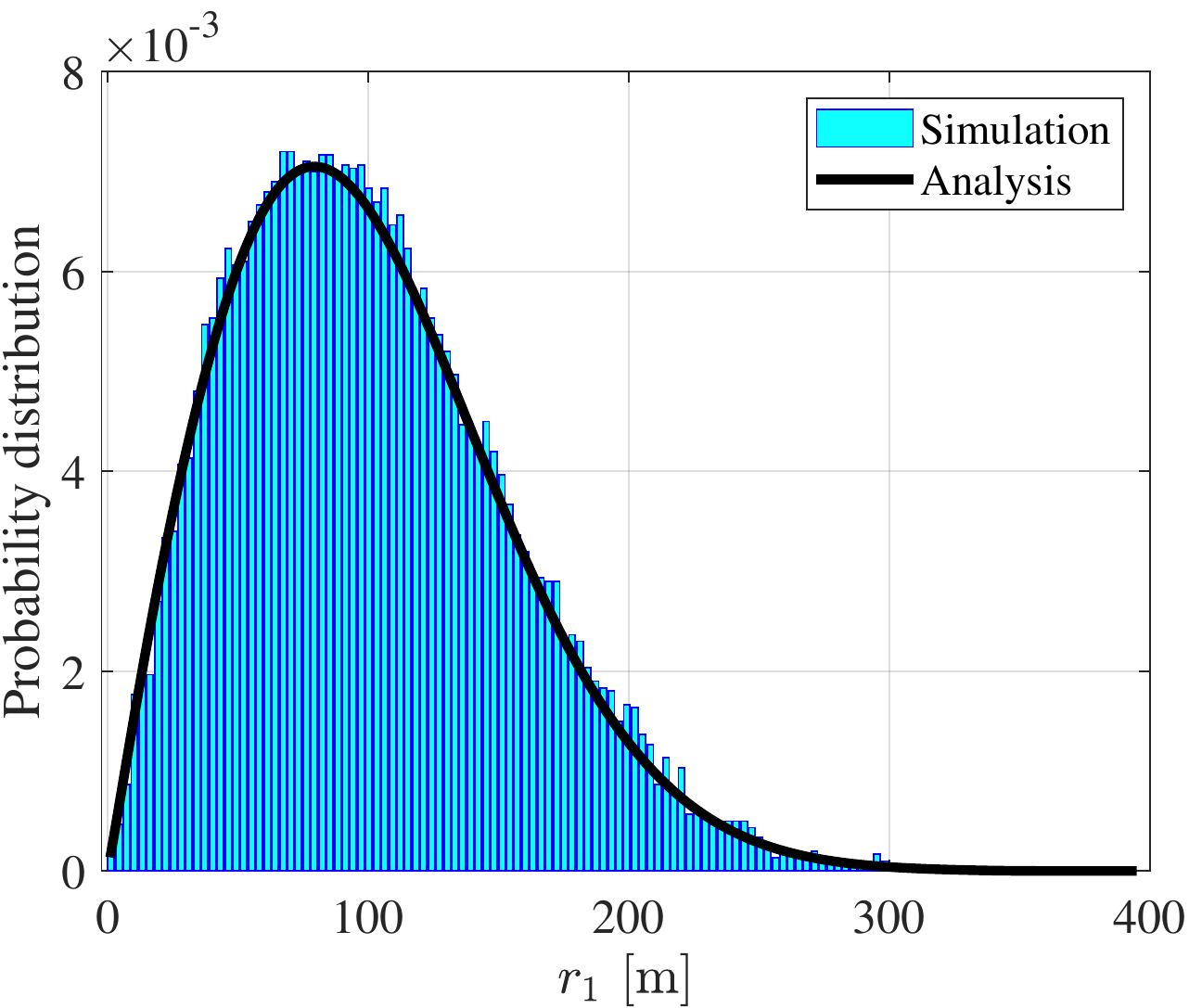}
                    \caption{Comparison of the analytical and simulation results for the PDF of $r_1$. $\lambda_{RIS}=1000 \frac{RIS}{km^2}$, $\lambda_{BS}=25 \frac{BS}{km^2}$. 
}
                \label{fig:r1sim}
\end{figure}
Fig. \ref{fig:r1sim} shows the theoretical PDF of the $r_1$ which coincides with the simulation results. 
In addition, in order to see the impact of $\lambda_{RIS}$ and $\lambda_{BS}$ on the $r_1$, the expected value of $r_1$ in \eqref{pdr10} with putting its condition into the integral becomes
\begin{align}
    \mathbb{E}\{r_{1}\}=\int_{r_0=0}^\infty \int_{r_{2}=0}^{\infty} &\int_{r_{1}=|r_0-r_{2}|}^{r_0+r_{2}}  r_{1} \times \nonumber\\
   & 
    f(r_{1}|r_{2},r_0) f(r_{2}) f(r_{0}) \,\,
    dr_{1}\,\, dr_{2}\,\, dr_0 \nonumber .
\end{align}
Thus, we have
\begin{align}
    \mathbb{E}\{r_{1}\}=&4 \pi \lambda_{BS}\lambda_{RIS}\int_{r_0=0}^\infty \int_{r_{2}=0}^{\infty} \int_{r_{1}=|r_0-r_{2}|}^{r_0+r_{2}} r_1^2 \,\, \times \nonumber \\
    &\frac{\exp \left[-\pi\left( \lambda_{RIS} r_2^2+\lambda_{BS} r_0^2\right)\right]}
    {\sqrt{1-\left(\frac{r_0^2+r_2^2-r_1^2}{2r_0r_2}\right)^2}}dr_{1}\,\, dr_{2}\,\, dr_0.
    \label{pro}
\end{align}
This integral is numerically calculated using MATLAB as a function of $\mathcal{F}_{r_1}\left(\lambda_{RIS},\lambda_{BS}\right)\triangleq \mathbb{E}\{r_1\}$ which Fig. \ref{r1Surf} shows its changes when $\lambda_{RIS}$ and $\lambda_{BS}$ vary. In general, an increase in either of $\lambda_{RIS}$ or $\lambda_{BS}$ reduces $\mathbb{E}\{r_1\}$.
 For example, for a fixed $\lambda_{BS}$, if $\lambda_{RIS_{a}}<\lambda_{RIS_{b}}$, then $\mathcal{F}_{r_1}\left(\lambda_{RIS_a},\lambda_{BS}\right)>\mathcal{F}_{r_1}\left(\lambda_{RIS_b},\lambda_{BS}\right)$. 
 However, the speed of changes in $\mathcal{F}_{r_1}$ decreases when $\lambda_{RIS}$ increases. Therefore, we can conclude that from law of large numbers (LLN) theorem, if $\lambda_{RIS_{a}}<\lambda_{RIS_{b}}$, $\lambda_{BS_{a}}<\lambda_{BS_{b}}$ and $\lambda_{BS_{a}},\lambda_{BS_{b}},\lambda_{RIS_{a}},\lambda_{RIS_{b}}\rightarrow \infty$, then $\mathcal{F}_{r_1}\left(\lambda_{RIS_a},\lambda_{BS_b}\right)\approx \mathcal{F}_{r_1}\left(\lambda_{RIS_b},\lambda_{BS_b}\right)$ in large-scale communications.

\subsection{Peak Reflection Signal Power of RIS}
The transmitted signal from the associated BS towards the engaged RIS experiences a small-scale fading gain, denoted by $f_m$, while reaching at the $m$th RIS-reflector, i.e., from \textit{Assumption 1}, we assume $f_m\sim \exp (\mu)$.
Thus, the impinging signal power at the $m$th passive RIS-reflector is given by
\begin{equation}
    \mathcal{P}_{m}=\frac{P_s}{2}f_m r_{1}^{-\alpha}.
\end{equation}
Since the RIS-reflectors are located in a relatively short distance from each other, we can assume $f_m$ is likely correlated and for a given $r_1$, we can conclude that 
\begin{equation}
   \mathcal{P}_m\approx \mathcal{P}_1=\frac{P_s}{2}f_1 r_1^{-\alpha}\quad \text{for all} \quad m.\label{pm}
\end{equation}

\begin{figure}[t]
                \centering
                    \includegraphics[width=9cm, height=6cm]{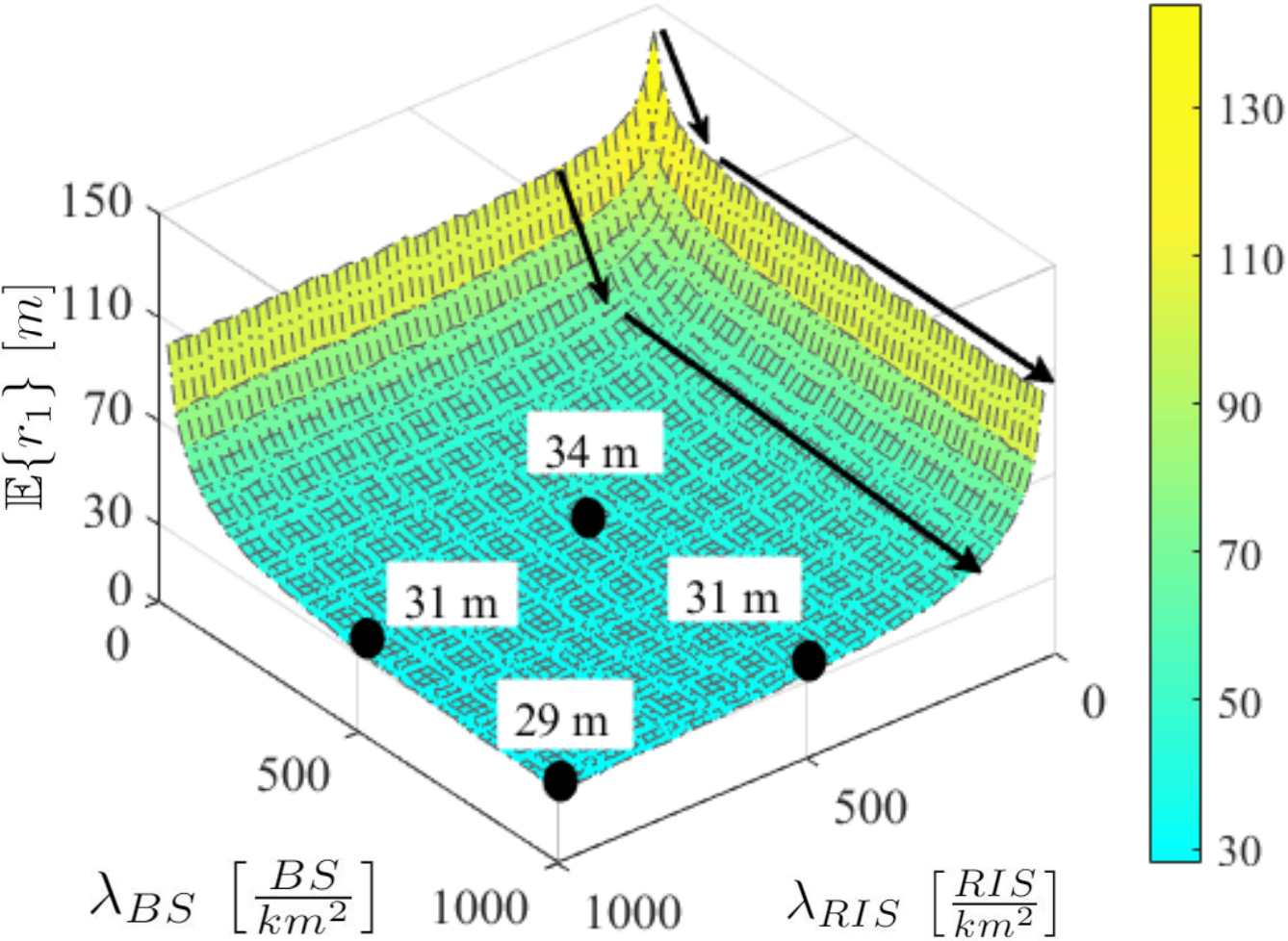}
                    \caption{Impact of $\lambda_{RIS}$ and $\lambda_{BS}$ on $\mathbb{E}\{r_1\}$. 
}
                \label{r1Surf}
\end{figure}
\begin{figure}[t]
                \centering
                    \includegraphics[width=9cm, height=6cm]{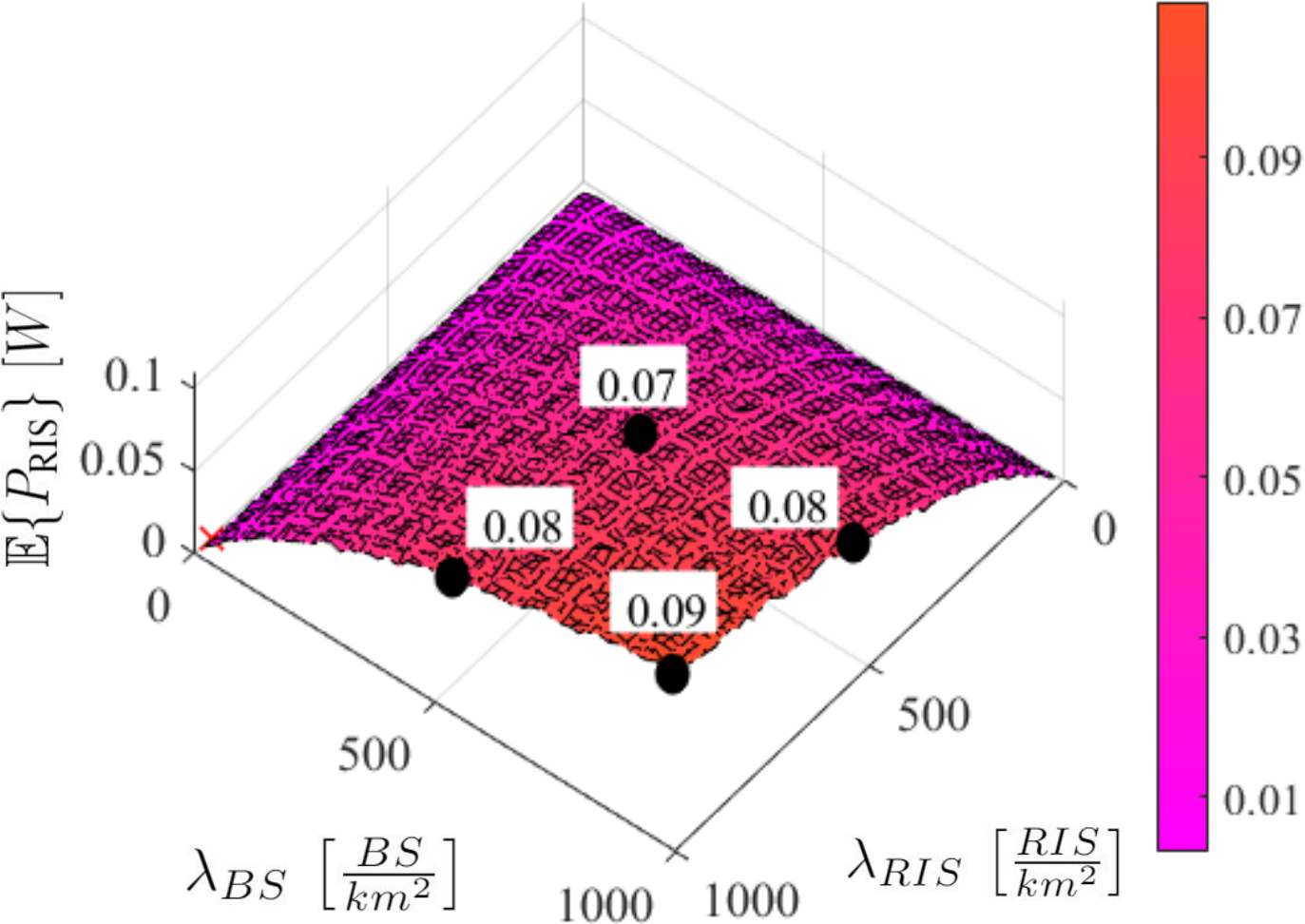}
                    \caption{Impact of $\lambda_{RIS}$ and $\lambda_{BS}$ on average reflected power, i.e., $\mathcal{F}_{P}\left(\lambda_{RIS},\lambda_{BS}\right)\triangleq \mathbb{E}\{P_\text{RIS}\}$, $M=100$, $\beta=1$, $\mu=1$, $\alpha=4$, and $P_s=2$ W.
}
                \label{power}
\end{figure}
Let us assume an attenuation power factor for RIS-reflectors which is denoted by $\beta \in (0, \quad 1]$. Here, $\beta$ is constant and can be obtained by measuring the attenuation power of a signal passing through the RIS-reflectors. Furthermore, since the peak effective radiated power scales up by $M^2$ as given in \eqref{trx1}, the peak reflected power from the engaged RIS towards the UE becomes 
\begin{equation}
    P_{\text{RIS}}=M^2  \beta \frac{P_s}{2} f_1 \,r_{1}^{-\alpha}.\label{pp1}
\end{equation}

Consequently the reflected power largely affected by $r_1^{-\alpha}$. With respect to \eqref{pro}, the average peak reflected power can be expressed as a function of $\lambda_{RIS}$ and $\lambda_{BS}$ for given $\alpha,\, \beta$ and $M$ as follows
\begin{align}
    \mathcal{F}_P\left(\lambda_{RIS},\lambda_{BS} \right)\triangleq \mathbb{E}\{P_\text{RIS}\}=
    \frac{M^2 \beta P_s}{2\mu}\mathbb{E}\{r_1^{-\alpha}\},\label{poavg}
\end{align}
 which is shown in Fig. \ref{power}. In other words, the dependency of average power on $\lambda_{RIS}$ and $\lambda_{BS}$ comes from $\tilde{\mathcal{F}}_{r_1}\left(\lambda_{RIS}, \lambda_{BS}\right)\triangleq \mathbb{E}\{r_1^{-\alpha}\}$ with respect to the PDF of $r_1$ in \eqref{pdr10} and \eqref{pro}.
 Thus, we conclude that in general, since $\alpha$ is positive, i.e., usually $\alpha>2$ \cite{RAP1}, the average reflected power from the RIS increases when $\lambda_{RIS}$ increases, i.e., $\mathbb{E}\{r_1\}$ decreases.
 For example, keeping $\lambda_{BS}$ fixed, if $\lambda_{RIS_{a}}<\lambda_{RIS_{b}}$, then $\mathcal{F}_{P}\left(\lambda_{RIS_a},\lambda_{BS}\right)<\mathcal{F}_{P}\left(\lambda_{RIS_b},\lambda_{BS}\right)$. 
Likewise, the reflected power from the RIS decreases when BS density decreases as shown in Fig. \ref{power}. However, this reflected power reduction can be compensated by employing larger RISs with more number of reflectors, i.e., a larger $M$.

In order to find the SIR coverage probability of the RIS-assisted path, i.e., path \textbf{B}, and simplify the analysis in section \ref{sec4}, we introduce a raw moment of $\left(\frac{P_{\text{RIS}}}{\mu}\right)^{\frac{2}{\alpha}}$ which is calculated as follows:
 \begin{align}
\mathbb{E}\left\{\left(\frac{P_{\text{RIS}}}{\mu}\right)^{\frac{2}{\alpha}}\right\}=\left[\frac{M^2\beta P_s}{2\mu}\right]^{\frac{2}{\alpha}}\mathbb{E}\left\{f_1^{\frac{2}{\alpha}} \right\}\mathbb{E}\left\{ r_1^{-2}\right\}.\label{EE}
\end{align}
 Since $f_1\sim \exp (\mu)$, we have
\begin{align}
    \mathbb{E}\left\{f_1^{\frac{2}{\alpha}} \right\}&=\int_{f_1=0}^{\infty} f_{1}^{\frac{2}{\alpha}} \mu e^{- \mu f_1}\, df_1=
    \mu^{-\frac{2}{\alpha}} \Gamma\left(\frac{2}{\alpha}+1\right),\label{f1}
\end{align}
where $\Gamma(x)=\int_0^\infty t^{x-1}e^-t \,dt$ is the Gamma function. Moreover, from \eqref{pdr10}, we define
\begin{align}
   \mathcal{F}_R\left(\lambda_{RIS},\lambda_{BS}\right)\triangleq \mathbb{E}\left\{ r_1^{-2}\right\}=\int_{r_{1}=\varepsilon}^\infty r_{1}^{-2}f_{r_{1}}(r_{1}) \,\, d r_{1}\label{r_12},
\end{align}
where $\varepsilon=\max \left(|r_0-r_2|,\epsilon\right)$ and $\epsilon$ is a minimum euclidean distance of the BS from the RIS that prevents the divergence of $\mathbb{E}\left\{r_{1}^{-2}\right\}$.
Eventually, by substituting \eqref{f1} and \eqref{r_12} in \eqref{EE}, we have
\begin{align}
\mathbb{E}\left\{\left(\frac{P_{\text{RIS}}}{\mu}\right)^{\frac{2}{\alpha}}\right\}
=
\left[\frac{M^2\beta P_s}{2\mu^2}\right]^{\frac{2}{\alpha}}
\Gamma\left(\frac{2}{\alpha}+1\right)
\mathcal{F}_R\left(\lambda_{RIS},\lambda_{BS}\right).
\label{EE1}
\end{align}


\textit{Remark 1:} The raw moment of $\left(\frac{P_{\text{RIS}}}{\mu}\right)^{\frac{2}{\alpha}}$ in \eqref{EE1} depends on $\mathcal{F}_R\left(\lambda_{RIS},\lambda_{BS}\right)$ for given $\alpha,\, M, \, \mu$ and $\beta$. Therefore, similar to the $\mathcal{F}_{P}\left(\lambda_{RIS},\lambda_{BS}\right)\triangleq \mathbb{E}\{P_\text{RIS}\}$ shown in Fig. \ref{power}, $\mathbb{E}\left\{\left(\frac{P_{\text{RIS}}}{\mu}\right)^{\frac{2}{\alpha}}\right\}$ is an increasing function of $\lambda_{RIS}$ and $\lambda_{BS}$.


%
 \section{SIR Coverage Probability in RIS-Assisted MmWave Cellular Network}
\label{sec4}
 In order to facilitate the UE implementation for the receptions through two paths \textbf{A} and \textbf{B} at the receiver, diversity is taken into account. By selection diversity of two received signals, the strongest signal is selected.
 Let $\Gamma_\textbf{A}$, and $\Gamma_\textbf{B}$ denote the SIRs of the received signals through paths \textbf{A} and \textbf{B}, respectively. Then, the SIR at the UE is given by
  \begin{equation}
      \Gamma_\textbf{s}=\max \{\Gamma_\textbf{A}, \Gamma_\textbf{B}\}.
      \label{gams}
 \end{equation}

 The SIR of the received signal from path \textbf{A} is similar to the baseline model while the beamwidth and transmit signal power at the active BS is being affected, i.e., $\psi_\textbf{s}=\sqrt{2}\psi_\textbf{o}$ and $\frac{P_s}{2}$ in 2D space. Consequently, the co-channel interference is affected as an additional gain of beamforming and the interferer BSs density 
 becomes
 \begin{equation}
     \lambda_{Is}=\sqrt{\frac{2}{N}}\lambda_{BS}.
     \label{lamI2}
 \end{equation}
 Therefore, taking these changes into account and similar to \eqref{q2}, the SIR coverage probability for path \textbf{A} can be obtained as
\begin{align}
    \Pr\left[\Gamma_\textbf{A}>T\right]=
    \frac{1}{\left(1+\sqrt{\frac{2}{N}}T^{\frac{2}{\alpha}}
         \int_{T^{-\frac{2}{\alpha}}}^\infty \frac{1}{1+u^{\frac{\alpha}{2}}}du\right)}
         \label{q23}
\end{align}     
 
 On the other hand, the SIR of the received signal through path \textbf{B} is given by
\begin{equation}
    \Gamma_\textbf{B}=\frac{P_{\text{RIS}} h r_2^{-\alpha}}{\sum\limits_{\substack{BS_i \in {\Phi}_{I},\\ i\neq 0}} \frac{P_s}{2} {g}_i r_i^{-\alpha}},
    \label{S-B}
\end{equation}
where $h$ denotes the small-scale fading gain between the engaged RIS and the UE. Based on \textit{Assumption 1} we have $h\sim \exp(\mu)$. In order to simplify the analysis to find a closed-form expression for $\Pr\left\{{ \Gamma_\textbf{B}}>T\right\}$, we utilize the following conversion.

\textit{Remark 2: (Power-Density Conversion)} The path-loss process of $r\in \Phi$ with transmit power $P$ and intensity $\lambda$ is equivalent with that of $\mathcal{R}\in \tilde{\Phi}$ with a transmit power of $1$ and an intensity of $\tilde{\lambda}$ given as
\begin{equation}
    \tilde{\lambda}=P^{\frac{2}{\alpha}}\lambda.
\end{equation}
The proof of this \textit{remark} is similar to that of \cite[lemma 1]{PLPL}. However, for the reader's convenience, a simplified version is provided in Appendix \ref{AP1}.
 
 Therefore based on \textit{remark 2}, we have $\Phi_{BS}\rightarrow \tilde{\Phi}_{BS}$, $\Phi_{I}\rightarrow \tilde{\Phi}_{I}$ and $\Phi_{RIS}\rightarrow \tilde{\Phi}_{RIS}$; and from \eqref{S-B}, the SIR coverage probability for path \textbf{B} is given by
 \begin{align}
    \Pr\left\{{ \Gamma_\textbf{B}}>T\right\}&=\mathbb{E}\left\{\Pr\left(\frac{\tilde{h} r_{2}^{-\alpha}}{\sum\limits_{\substack{BS_i \in \tilde{\Phi}_{I},\\ i\neq 0}} \tilde{g}_i r_i^{-\alpha}}>T\right)\right\},
    \label{S-B2}
    \end{align}
    where 
    \begin{align}
        &\tilde{\lambda}_{BS}= \left(\frac{P_s}{2\mu} \right)^{\frac{2}{\alpha}} \lambda_{BS} \label{uy1}\\
        &
        \tilde{\lambda}_I= \left ( \frac{P_s}{2\mu}\right)^{\frac{2}{\alpha}}\lambda_{Is}\\
        &
        \tilde{\lambda}_{RIS}=\mathbb{E}\left\{\left(\frac{P_{\text{RIS}}}{\mu}\right)^{\frac{2}{\alpha}}\right\}\lambda_{RIS},\label{we23}\\
        &
        \tilde{g}_i\sim \exp (1) \hspace{2mm}\text{   and   }\hspace{2mm}
        \tilde{h}\sim \exp (1) .\label{we234}
    \end{align}

By taking the changes in \eqref{uy1} to \eqref{we234} into account, the SIR coverage probability in \eqref{S-B2} becomes
\begin{align}
    \Pr\left\{{ \Gamma_\textbf{B}}>T\right\}&=
     \mathbb{E}\left\{\Pr\left(\tilde{h}>Tr_{2}^{\alpha}\sum\limits_{BS_i\in\tilde{\Phi}_{I}} \tilde{g}_i r_i^{-\alpha}\right)\right\}.\label{tre2}
\end{align}
Since $\tilde{h}\sim \exp (1)$ and $\Pr(\tilde{h}>x)=1-F_{\tilde{h}}(x)=e^{-x}$, \eqref{tre2} becomes
\begin{align}
 \Pr\left\{{\Gamma_\textbf{B}}>T\right\}&=\mathbb{E}\left\{\exp \left(-Tr_{2}^{\alpha}\sum\limits_{BS_i\in\tilde{\Phi}_{I}}\tilde{g}_ir_i^{-\alpha}\right)\right\}\nonumber\\
 &=
    \mathbb{E}\left\{\prod\limits_{BS_i\in\tilde{\Phi}_{I}}\exp \left(- T \tilde{g}_i \left[\frac{r_{2}}{r_i}\right]^{\alpha}\right)\right\}.\label{e1}
\end{align}
Since $\tilde{g}_i\sim \exp (1)$, we have $\mathbb{E}\left\{e^{-\tilde{g}_ix}\right\}=\left[1+x\right]^{-1}$ and \eqref{e1} becomes
\begin{align}
    \Pr\left\{{\Gamma_\textbf{B}}>T\right\}&= \mathbb{E}\left\{\prod\limits_{BS_i\in\tilde{\Phi}_{I}} \frac{1}{1+T  \left[\frac{r_{2}}{r_i}\right]^\alpha}\right\}.\label{e3}
\end{align}
With simplification by application of Campbell’s theorem \cite{parkBW,Champ,REF}\footnote{$\mathbb{E}\left\{\prod\limits_{\substack{x \in \tilde{\Phi}_{I}}} f(x)\right\}=\exp \left[-\tilde{\lambda}_{I}\int_{\mathbb{R}^2}(1-f(x))dx\right].$}, from \eqref{e3} we have
\begin{align}
    &\mathbb{E}\left\{\prod\limits_{BS_i\in\tilde{\Phi}_{I}} \frac{1}{1+T  \left[\frac{r_{2}}{r_i}\right]^\alpha}\right\}=\nonumber\\
&       
\mathbb{E}\left\{\exp \left(-2\pi\tilde{\lambda}_{I} \int_{\upsilon=r_0}^\infty \left[1-\frac{1}{1+T  \left[\frac{r_{2}}{\upsilon}\right]^\alpha}\right]\upsilon d\upsilon\right)\right\}.\label{ew320}
\end{align}
Simplifying the integral expression in \eqref{ew320} is a complicated task since the lower bound in integral is a function of $r_0$ and the expression inside the integral is a function of $r_2$ which are independent of each other. To simplify \eqref{ew320} and find a closed-form expression for $\Gamma_{\textbf{B}}$ coverage probability, at this stage, we provide two approximations as follows.

\subsection{\textbf{Approximation-\rom{1}}} In this approximation, from \eqref{pdfR0} and \eqref{pdfr2}, we have
\begin{equation}
    \mathbb{E}\{r_{2,0}\}=\sqrt{\frac{\tilde{\lambda}_{BS}}{\tilde{\lambda}_{RIS}}}\mathbb{E}\{r_0\}
\end{equation}
Therefore, we simply approximate
\begin{equation}
    r_2\approx \rho r_0,
\label{assu1}
\end{equation} 
where $\rho=\sqrt{\frac{\tilde{\lambda}_{BS}}{\tilde{\lambda}_{RIS}}}$. Consequently we define the following \textit{proposition}.

\begin{proposition}
The SIR coverage probability for path \textbf{B} of a randomly located UE in RIS-assisted mmWave cellular network with \textbf{Approximation \rom{1}} is \begingroup\makeatletter\def\f@size{9.5}\check@mathfonts
    \begin{align}
    Pr\left\{{\Gamma_\textbf{B}}>T\right\}&= \frac{\tilde{\lambda}_{RIS}}{\tilde{\lambda}_{RIS}
    +
    \frac{\tilde{\lambda}_{I}}{\rho^2}\left(T\right)^{\frac{2}{\alpha}} \int_{u=\left(T\right)^{-\frac{2}{\alpha}}}^\infty \frac{\rho^\alpha}{\rho^\alpha+ u^{\frac{\alpha}{2}}} du}\label{er431}
\end{align}
\endgroup
\end{proposition}

\begin{proof}
 By Taking \eqref{assu1} into account, \eqref{ew320} becomes
\begin{align}
    \mathbb{E}\left\{\exp \left(-2\pi\tilde{\lambda}_{I} \int_{\upsilon=\frac{r_2}{\rho}}^\infty \left[1-\frac{1}{1+T  \left[\frac{r_{2}}{\upsilon}\right]^\alpha}\right]\upsilon d\upsilon\right)\right\}
\end{align}
By employing a change of variable $u= \frac{\rho^2\upsilon^2}{r_{2}^2\left(T\right)^{\frac{2}{\alpha}}}$, we have
\begin{equation}
      \mathbb{E}\left\{\exp \left(-\frac{\pi\tilde{\lambda}_{I}\left(T\right)^{\frac{2}{\alpha}}r_2^2}{\rho^2}
       \int_{u=\left(T\right)^{-\frac{2}{\alpha}}}^\infty \frac{\rho^\alpha}{\rho^\alpha+ u^{\frac{\alpha}{2}}} du\right)\right\}\label{trq}
\end{equation}
By taking the average of \eqref{trq} over $r_2$ with respect to \eqref{pdfr2} and \eqref{we23}, we have
\begin{align}
    \int_{r_{2}=0}^{\infty} f_{r_2}(r_2) \exp &\left(-
    \frac{\pi\tilde{\lambda}_{I}\left(T\right)^{\frac{2}{\alpha}}r_2^2}{\rho^2}\times
       \right.\nonumber\\
       &\hspace{9mm}\left.\int_{u=\left(T\right)^{-\frac{2}{\alpha}}}^\infty \frac{\rho^\alpha}{\rho^\alpha+ u^{\frac{\alpha}{2}}} du\right)\, dr_{2}.\label{e5}
\end{align}
Let consider parameter $\mathcal{J}$ as follows
\begin{align}
     \mathcal{J}=\pi\left(\tilde{\lambda}_{RIS}
    +
    \frac{\tilde{\lambda}_{I}}{\rho^2}\left(T\right)^{\frac{2}{\alpha}} \int_{u=\left(T\right)^{-\frac{2}{\alpha}}}^\infty \frac{\rho^\alpha}{\rho^\alpha+ u^{\frac{\alpha}{2}}} du
    \right).\label{AA}
\end{align}
Then, \eqref{e5} becomes
\begin{align}
    \Pr\left\{{\Gamma_\textbf{B}}>T\right\}=
    2\pi\tilde{\lambda}_{RIS} 
    \int_{r_{2}=0}^{\infty} r_{2} e^{-\mathcal{J} r_{2}^2}\, dr_{2}=  
   \frac{\pi\tilde{\lambda}_{RIS}}{\mathcal{J}}\label{e6}
\end{align}
Eventually by substituting $\mathcal{J}$ in \eqref{AA} into \eqref{e6} we have
\begingroup\makeatletter\def\f@size{9.5}\check@mathfonts
    \begin{align}
    Pr\left\{{\Gamma_\textbf{B}}>T\right\}&= \frac{\tilde{\lambda}_{RIS}}{\tilde{\lambda}_{RIS}
    +
    \frac{\tilde{\lambda}_{I}}{\rho^2}\left(T\right)^{\frac{2}{\alpha}} \int_{u=\left(T\right)^{-\frac{2}{\alpha}}}^\infty \frac{\rho^\alpha}{\rho^\alpha+ u^{\frac{\alpha}{2}}} du}\nonumber
\end{align}
\endgroup
\end{proof}
\subsection{\textbf{Approximation-\rom{2}}} In this approximation we consider a lower bound where $\lambda_{RIS}\gg\lambda_{BS}\rightarrow \upsilon\geqslant r_{2}$ in \eqref{ew320}. Therefore, we state the following \textit{proposition}.
\begin{proposition}
The lower bound SIR coverage probability for path \textbf{B} of a randomly located UE in RIS-assisted mmWave cellular network when $\lambda_{RIS}\gg\lambda_{BS}$ with \textbf{Approximation \rom{2}} is
\begin{align}
    \Pr\left\{{\Gamma'_\textbf{B}}>T\right\}=& \frac{\tilde{\lambda}_{RIS}}{\tilde{\lambda}_{RIS}
    +
    \tilde{\lambda}_{I}\left(T\right)^{\frac{2}{\alpha}} \int_{u=\left(T\right)^{-\frac{2}{\alpha}}}^\infty \frac{1}{1+ u^{\frac{\alpha}{2}}} du}\label{er43}
\end{align}
\end{proposition}

\begin{proof}
 Since the proof is similar to that in \textit{proposition 1}, we
omit it here. 
\end{proof}

 \section{Discussion on Baseline and RIS-Assisted MmWave Cellular Networks}\label{discus}
 
 The SIR coverage probability of the baseline model in \eqref{q2} and path \textbf{A} in the RIS-assisted model in \eqref{q23} does not depend on $\lambda_{BS}$; and the only deployment parameter to enhance the SIR performance is $N$ which means that with $N\rightarrow \infty$, we have
 \begin{equation}
 N\rightarrow \infty \Rightarrow
 \left\{\begin{array}{l}
      \Pr\left\{{\Gamma_\textbf{o}}>T\right\}\rightarrow 1  \\
      \Pr\left\{{\Gamma_\textbf{A}}>T\right\}\rightarrow 1
 \end{array} \right. .
 \end{equation}
 However the lower bound SIR probability in \eqref{er43} in \textbf{Approximation \rom{2}} shows that the SIR performance for path \textbf{B} depends on deployment parameters of not only $N$ but also $\lambda_{BS}$, $\lambda_{RIS}$, and $M$. In other words, from \eqref{EE1} and \eqref{we23}, \eqref{er43} can be re-expressed as
 \begin{align}
     &\Pr\left\{{\Gamma'_\textbf{B}}>T\right\}=\nonumber\\
    &\hspace{1mm} \frac{{\lambda}_{RIS} M^{\frac{4}{\alpha}} \mathcal{F}_1\left(\lambda_{BS},\lambda_{RIS},\alpha,\beta\right)
     }
     {
     {\lambda}_{RIS} M^{\frac{4}{\alpha}} \mathcal{F}_1\left(\lambda_{BS},\lambda_{RIS},\alpha,\beta\right)
      + \sqrt{\frac{2}{N}}{\lambda}_{BS}\mathcal{F}_2(T,\alpha)
     },\label{52}
 \end{align}
 where
 \begin{align}\left\{
 \begin{array}{l}
     \mathcal{F}_1\left(\lambda_{BS},\lambda_{RIS},\alpha,\beta\right)=\left[\frac{\beta}{\mu}\right]^{\frac{2}{\alpha}}
\Gamma\left(\frac{2}{\alpha}+1\right)\mathcal{F}_R\left(\lambda_{RIS},\lambda_{BS}\right) \vspace{1mm} \\
      \mathcal{F}_2(T,\alpha)=\left(T\right)^{\frac{2}{\alpha}} \int_{u=\left(T\right)^{-\frac{2}{\alpha}}}^\infty \frac{1}{1+ u^{\frac{\alpha}{2}}} du 
 \end{array}
 \right. .\nonumber
 \end{align}
 
 Based on \textit{remark 1}, $\mathcal{F}_R$ and consequently $\mathcal{F}_1$ are increasing functions of $\lambda_{RIS}$. Therefore, we have
 \begin{align}
 \left.
     \begin{array}{l}
        \hspace{1mm} \lambda_{RIS}\rightarrow\infty \\
      or \hspace{2mm} M\rightarrow \infty \\
      or  \hspace{2mm} N\rightarrow \infty
     \end{array}\right\}\Rightarrow \Pr\left\{{\Gamma'_\textbf{B}}>T\right\}\rightarrow 1.
 \end{align}
Nevertheless, in practice, there are implementation issues which might prevent these parameters to become very large. However, there is a great deal of flexibility to select a proper parameter for SIR enhancement. For instance, a lower complex antenna array, i.e., smaller $N$, can be deployed at the BS and either higher RIS density or larger number of reflectors $M$ can be taken into account to provide a desired SIR gain.
On the other hand, from \eqref{52} when $\lambda_{BS}$ increases, although the $\mathcal{F}_1$ in the numerator increases, the expression in the denominator increases faster than the numerator because of additive expression of $\sqrt{\frac{2}{N}}{\lambda}_{BS}\mathcal{F}_2(T,\alpha)$. In other words, the co-channel interference increases faster than the reflected power from the RIS when the active BS density increases, i.e., 
\begin{align}
     \Pr\left\{{\Gamma'_\textbf{B}}>T\right\}\varpropto \frac{C_1\lambda_{RIS}}{C_1\lambda_{RIS}+C_2\lambda_{BS}},
     \end{align}
     where $C_1$ and $C_2$ are assumed to be constant for specific moments.
     Therefore, it is desirable to decrease the number of active BSs ($\lambda_{BS}\neq 0$) and deploy more passive RISs.




\section{Numerical Results}
In this section, the performance of the RIS-assisted is shown by simulation results. The evaluations include: A) the SIR coverage probability comparison for both paths \textbf{A} and \textbf{B} in the RIS-assisted model, B) the SIR coverage probability comparison of the RIS-assisted and the baseline models, and C) the impact of $\lambda_{RIS}$ and $\lambda_{BS}$ on SIR coverage probability.


We use MATLAB and the parameters used in the simulations are given in TABLE \ref{tab1}, unless otherwise specified.
\begin{table}[h]
\begin{center}
\centering
\captionsetup{width=1\linewidth}
\caption{System numerical parameters.
\label{tab1}}
\label{tab:par}
\begin{tabular}{lc} \hline
\hline
System parameters & Corresponding value  \\ 
 \hline
 \hline
BS radiated power (downlink) , $P_s$  & 2 W\\
 \hline
BS antenna array, $N$ & $16$ antennas \\
\hline
RIS elements, $M$ & 100 reflectors\\
\hline
Attenuation power ratio at RIS-elements, $\beta$ & 0.9\\
\hline
Channel bandwidth, $W$  & $100$ MHz\\
\hline
\hline
\end{tabular}
\end{center}
\end{table}

\subsection{SIR Coverage Evaluation in RIS-assisted Model}
\begin{figure}[t]
\centering
\captionsetup{width=1\linewidth}
  \includegraphics[width=9cm, height=6cm]{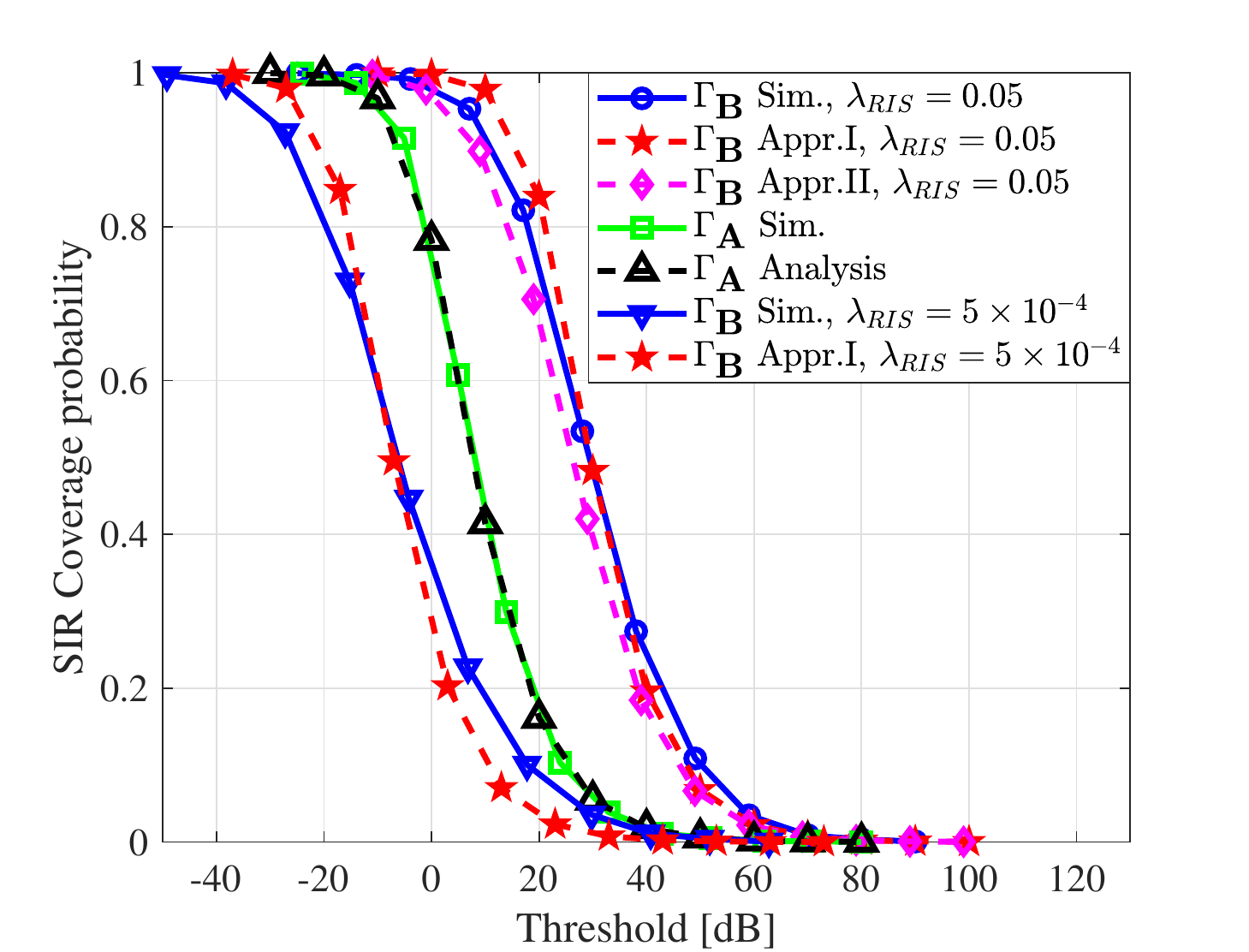}%
  \caption{Coverage evaluation of the RIS-assisted model through paths \textbf{A} and \textbf{B}. RIS densities are in $\frac{RIS}{m^2}$, $\lambda_{BS}=2.5\times 10^{-5}\frac{BS}{m^2}$, and $M=100$. }
  \label{sim1}
\end{figure}
Since the selection diversity of the two received signals is taken into account in the proposed RIS-assisted model, the UE is able to select the strongest signal between those received signals through paths \textbf{A} and \textbf{B}. Fig. \ref{sim1} compares the SIR coverage probabilities for both of paths $\textbf{A}$ and \textbf{B}. Based on \eqref{q23}, \eqref{er431} and \eqref{er43}, $\Gamma_\textbf{A}$ only depends on $N$ while $\Gamma_\textbf{B}$ depends on $\lambda_{RIS}$ as well. Therefore, it is shown that when $\lambda_{RIS}$ increases from $5\times 10^{-4}$ to $0.05$ $\frac{RIS}{m^2}$, $\Gamma_\textbf{B}$ is improved. Moreover, when $\lambda_{RIS}$ becomes larger, i.e., $\lambda_{RIS}\gg\lambda_{BS}$, $\Gamma_\textbf{B}$ outperforms $\Gamma_\textbf{A}$.

\subsection{SIR Comparison of RIS-assisted and Baseline Models}
Fig. \ref{sim2} compares the the coverage performance of the proposed model with the baseline model. It shows that there is a minimum coverage probability when there is no RIS-assisted path between the BS and the UE, i.e., small $\lambda_{RIS}$ results in $\Gamma_\textbf{s}=\Gamma_\textbf{A}$ which is independent of $\lambda_{RIS}$, e.g., $\lambda_{RIS}=10^{-5}$ $\frac{RIS}{m^2}$. This minimum coverage probability is slightly less than that of the baseline model since we have two beams in the RIS assisted model while there is only one beam in the baseline model with the same BS structure. In other words, with fixed $N$ active antennas at the BSs, the beam in the baseline model is narrower that that of the RIS-assisted model causes less co-channel interference. However, when $\lambda_{RIS}$ increases, the coverage probability in the RIS-assisted is improved. In addition, bigger RISs, i.e., a larger $M$, performs better than small RISs as it increases the reflected power from the engaged RIS.
\begin{figure}[t]
\centering
\captionsetup{width=1\linewidth}
  \includegraphics[width=9cm, height=6cm]{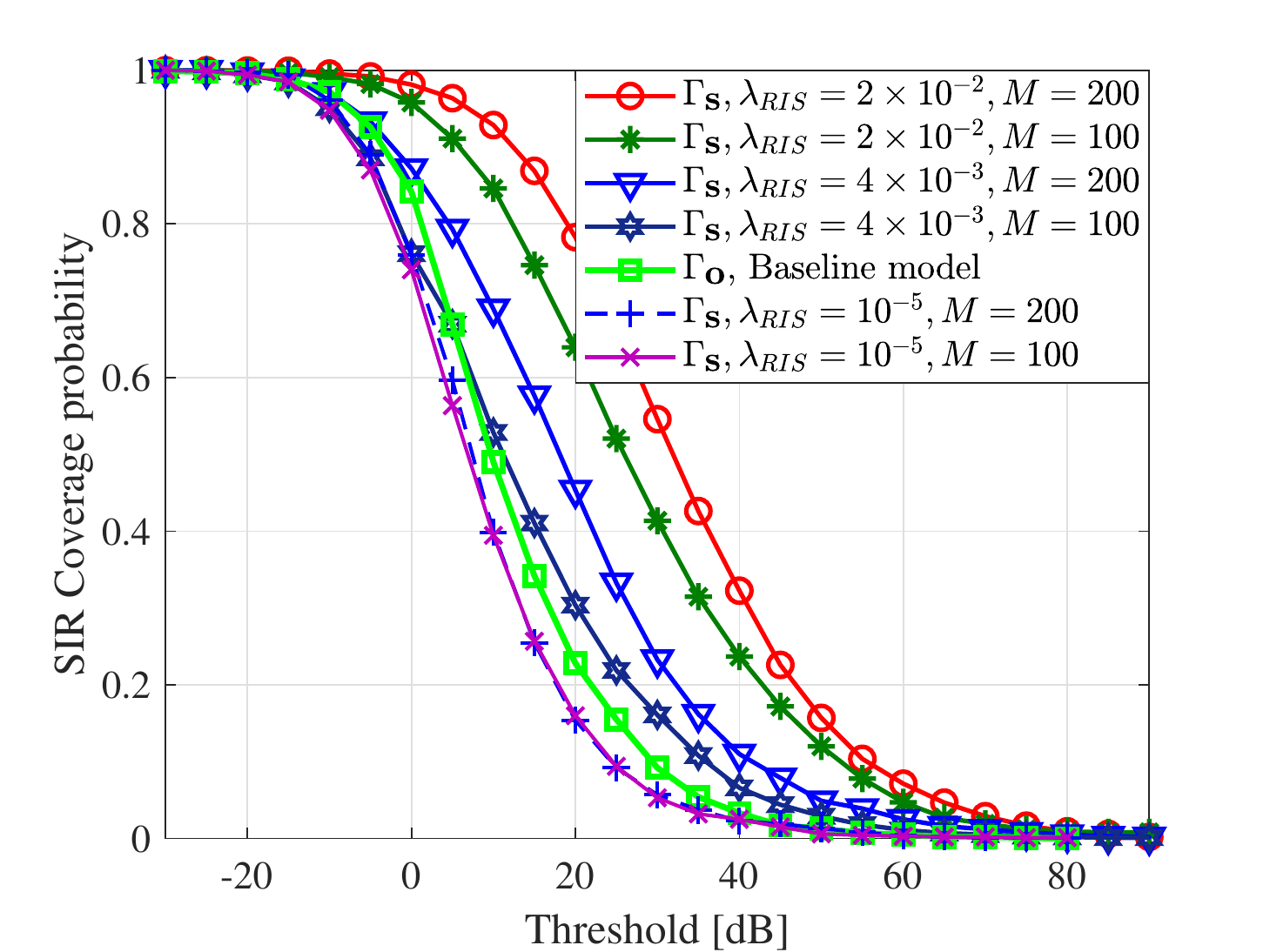}%
  \caption{Coverage comparison between the RIS-assisted and the baseline models. RIS densities are in $\frac{RIS}{m^2}$. $\lambda_{BS}=2.5\times 10^{-5}\frac{BS}{m^2}$. }
  \label{sim2}
\end{figure}
\subsection{Impact of $\lambda_{RIS}$ and $\lambda_{BS}$ on SIR Coverage Performance}
\begin{figure}[t]
\centering
\captionsetup{width=1\linewidth}
  \includegraphics[width=9cm, height=6cm]{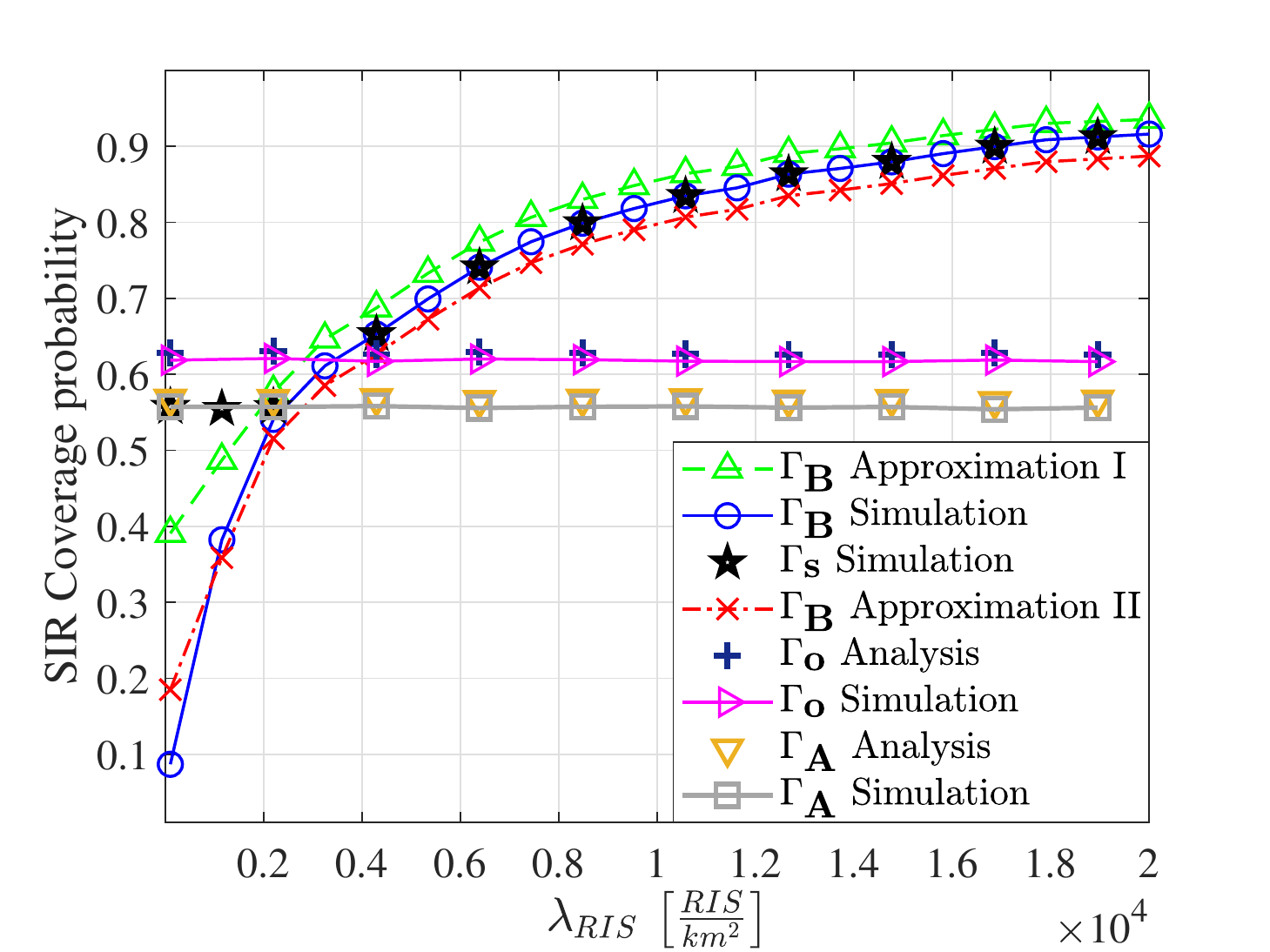}%
  \caption{Impact of $\lambda_{RIS}$ on SIR coverage performance. $\lambda_{BS}=2.5\times 10^{-5}\frac{BS}{m^2}$, $M=100$ and $T=5$ dB. }
  \label{3sim}
\end{figure}
Fig. \ref{3sim} shows the impact of $\lambda_{RIS}$ on the proposed RIS-assisted model. As discussed in section \ref{discus}, when $\lambda_{RIS}$ increases, the SIR coverage probability of $\Gamma_\textbf{B}$ approaches 1. 
Moreover, we can see that \textbf{Approximation \rom{1}} is supported by the simulation results in high $\lambda_{RIS}$ because when $\lambda_{RIS}\rightarrow \infty$, the distance of $r_2$ and its variations becomes smaller; therefore, $r_2\approx \rho r_0$ in \eqref{assu1} becomes more realistic. 
In addition, the SIR performances for the baseline model and the path \textbf{A} in the RIS-assisted model are fixed and independent from the $\lambda_{RIS}$ variations. $\Gamma_\textbf{o}$ has better performance compared to $\Gamma_{A}$ due to having a narrower beamwidth since it uses all $N$ elements to create one beam while in the proposed model there are two wider beams with the same BS structure.
\begin{figure}[t]
\centering
\captionsetup{width=1\linewidth}
  \includegraphics[width=9cm, height=6cm]{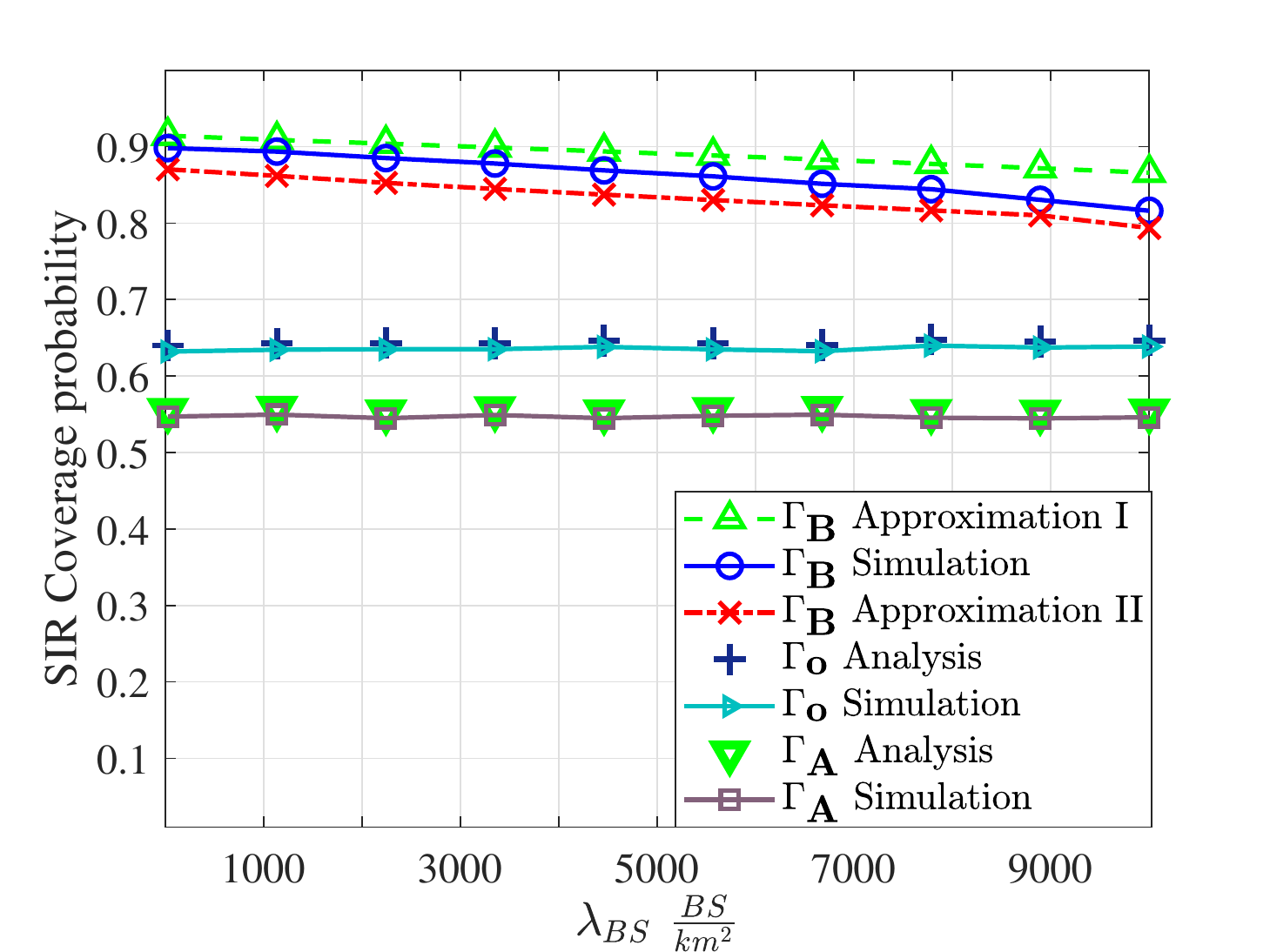}%
  \caption{Impact of $\lambda_{BS}$ on SIR coverage performance. $\lambda_{RIS}=2\times 10^{-2}\frac{RIS}{m^2}$ and $T=5$ dB. }
  \label{4sim}
\end{figure}

Fig. \ref{4sim} depicts the SIR coverage performance of the RIS-assisted and the baseline model with respect to $\lambda_{BS}$. As $\lambda_{BS}$ increases, $\Gamma_\textbf{o}$ and $\Gamma_\textbf{A}$ remain unchanged while $\Gamma_\textbf{B}$ slightly decreases. It shows that although higher $\lambda_{BS}$ leads to a higher reflected power from the RIS, its impact on the co-channel interference power is larger. In other words, the speed of increasing of co-channel interference power is faster than that of reflected power from the RIS.    

\section{Conclusions}
In this study, we proposed a new RIS-assisted mmWave cellular network where a message is transmitted by a BS towards a desired UE though two LoS and NLoS paths. The NLoS path passes through an RIS and then, reflected towards the UE. Discrete time delay values corresponding the phase-shifts at each RIS-reflector was elaborated and the peak reflection power at the RIS was assessed.
Since the UE utilizes selection diversity technique to pick the strongest signal received through the two paths, we analysed the SIR coverage performance of both paths with major emphasis on RIS and BS densities and compared its performance with a baseline model. Two closed-form approximations were derived analytically for the SIR coverage probability of the RIS-assisted path. It was shown that the SIR coverage probability of the RIS-assisted path depends on not only $N$ but also $\lambda_{RIS}$, $\lambda_{BS}$, and $M$ which provides a great deal of flexibility to obtain a desired SIR gain. 

\section{Acknowledgement}
This work was supported by Australian Research Council (ARC) Discovery 2020 Funding, under grant number DP200100391.

\begin{appendices}
\section{Derivation of $\Pr\left[\Gamma_\textbf{o}>T\right]$}
\label{AP0}

Similar to the analysis in \cite{REF}, the SIR coverage probability in \eqref{eq213} becomes
\begin{align}
    &\Pr\left[\Gamma_\textbf{o}>T\right]=
    \mathop{\mathbb{E}}\limits_{\substack{{g}_0\\{g}_i\\r_i\\r_0}}\left\{\Pr\left({g}_0> Tr_0^{\alpha}\sum\limits_{\substack{BS_i \in {\Phi}_{I}\\i\neq 0}} {g}_i r_i^{-\alpha}\right)\right\}=\nonumber\\
%
    &
    \mathop{\mathbb{E}}\limits_{\substack{{g}_i\\r_i\\r_0}}\left\{\exp\left(-\mu Tr_0^{\alpha}\sum\limits_{\substack{BS_i \in {\Phi}_{I}\\i\neq 0}} g_i r_i^{-\alpha}\right)\right\}=\nonumber\\
    &
    \mathop{\mathbb{E}}\limits_{\substack{r_i\\r_0}}\left\{ \prod\limits_{\substack{BS_i \in {\Phi}_{I}\\i\neq 0}} \mathop{\mathbb{E}}\limits_{\substack{{g}_i}}\left\{\exp \left( -\mu {g}_iT\left[\frac{r_0}{r_i}\right]^\alpha\right)\right\}\right\}\stackrel{\mathbb{E}\left\{e^{-{g}_ix}\right\}=\mu\left[\mu+x\right]^{-1}}{=}\nonumber\\
    &
    \mathop{\mathbb{E}}\limits_{\substack{r_0}}\left\{ 
    \mathop{\mathbb{E}}\limits_{\substack{r_i}}\left\{\prod\limits_{\substack{BS_i \in {\Phi}_{I}\\i\neq 0}} \mu\left(\mu+\mu T\left[\frac{r_0}{r_i}\right]^\alpha\right)^{-1}\right\}\right\}\stackrel{\text{Campbell's theorem}}{=}\nonumber\\
    &
    \mathop{\mathbb{E}}\limits_{\substack{r_0}}\left\{ 
    \exp \left(-2\pi{\lambda}_{I}
    \int_{\upsilon=r_0}^\infty \left[ 1-\left(1+T\left[\frac{r_0}{\upsilon}\right]^\alpha\right)^{-1}\right]\upsilon d\upsilon\right) \right\}. 
    \label{eq218}
\end{align}
With changing the variable as $
    u= \frac{\upsilon^2}{r_{0}^2\left(T\right)^{\frac{2}{\alpha}}}
$ and averaging over $r_0$,
we have
\begingroup\makeatletter\def\f@size{10}\check@mathfonts
\begin{align}
    &\Pr\left[\Gamma_\textbf{o}>T\right]=\nonumber\\
    &
    \int_{r_0=0}^\infty f_{r_o}(r_o)\exp{\left(-\pi{\lambda}_{I}r_{0}^2\left(T\right)^{\frac{2}{\alpha}}
    \int_{u=\left(T\right)^{-\frac{2}{\alpha}}}^\infty \frac{1}{1+ u^{\frac{\alpha}{2}}} du\right)}d \, r_0.
    \label{eq220}
\end{align}
\endgroup
Eventually, the closed form of \eqref{eq220} becomes
\begin{align}
        \Pr\left[\Gamma_\textbf{o}>T\right]&=
    \frac{{\lambda}_{BS}}{\left({\lambda}_{BS}+{\lambda}_{I}T^{\frac{2}{\alpha}}
         \int_{T^{-\frac{2}{\alpha}}}^\infty \frac{1}{1+u^{\frac{\alpha}{2}}}du\right)}\\
         &\stackrel{\lambda_I=\frac{\lambda_{BS}}{\sqrt{N}}}{=}
         \frac{1}{\left(1+\frac{1}{\sqrt{N}}T^{\frac{2}{\alpha}}
         \int_{T^{-\frac{2}{\alpha}}}^\infty \frac{1}{1+u^{\frac{\alpha}{2}}}du\right)}.\nonumber 
\end{align}

\section{Derivation of $f_{r_2}(r|r_2<r_0)$}
Since, the BSs and RISs are independently distributed in the area and utilizing Bayes theorem,
the CDF of being $r_2<r_0$ is given by
\begin{equation}
    F_{r_2}\left(R|r_2<r_0\right)=\frac{\int_{r_2=0}^R 
    \Pr(r_0>r_2|r_2) f_{r_2}({r_2})  d r_2}{F_{r_2}(r_0)}.
    \label{ap11}
\end{equation}
From \eqref{pdfR0} and \eqref{pdfr2}, the numerator of \eqref{ap11} becomes
\begin{align}
&
\int_{r_2=0}^{R} \left[e^{-\pi\lambda_{BS}r_2^2}\right]f_{r_2}(r_2) d r_2=\nonumber \\
&
\hspace{2cm}
\frac{\lambda_{RIS}}{\lambda_{BS}+\lambda_{RIS}}\left[1-e^{-\pi(\lambda_{RIS}+\lambda_{BS})R^2}\right].
\label{lt32}
\end{align}
Then, the denominator of \eqref{ap11} becomes 
\begin{align}
   F_{r_2}(r_0)&= \Pr \left[r_{2}<r_0\right]\nonumber\\
         &= 
         \int_{r_0=0}^\infty \left[\int_{r_{2}=0}^{r_0} \,f_{r_{2}}(r_{2}) \,dr_{2}\right] f_{r_0}(r_0) dr_0\nonumber\\
         &= \frac{\lambda_{RIS}}{\lambda_{BS}+\lambda_{RIS}}.\label{lr32}
\end{align}
Eventually, by substituting \eqref{lt32} and \eqref{lr32} into \eqref{ap11}, we have \begin{equation}
    F_{r_2}\left(R|r_2<r_0\right)=1-e^{-\pi(\lambda_{RIS}+\lambda_{BS})R^2}.
    \label{ap13}
\end{equation}
Consequently, we have
\begin{align}
f_{r_2}(r|r_2<r_0)=2\pi(\lambda_{RIS}+\lambda_{BS}) r e^{-\pi(\lambda_{RIS}+\lambda_{BS})r^2}. \nonumber
\end{align}
The analysis is complete.
\label{AP00}

\section{Power-Density Conversion}
Let $r=\sqrt{x^2+y^2}$ and $x,y\in \Phi$ (homogeneous PPP) with intensity of $\lambda$ in the 2-dimensional Euclidean plane. Similarly, let also consider a new r.v. $\mathcal{R}=\sqrt{\mathcal{X}^2+\mathcal{Y}^2}$ where $\mathcal{X},\mathcal{Y}\in \tilde{\Phi}$ with intensity of $\tilde{\lambda}$.
Suppose constant values of $P$ and $\alpha$ where we have
\begin{align}
    \mathcal{R}^{-\alpha}=P r^{-\alpha}&=\left[\left(P\right)^{-\frac{1}{\alpha}} r\right]^{-\alpha}\nonumber\\
    &=\left[\sqrt{
      \underbrace{\left(P\right)^{-\frac{2}{\alpha}} x^2}_{\mathcal{X}}+\underbrace{\left(P\right)^{-\frac{2}{\alpha}} y^2}_{\mathcal{Y}}  }\right]^{-\alpha}.
\end{align}
Therefore, we have
\begin{equation}
    \left[\begin{array}{c}
         \mathcal{X}  \\
         \mathcal{Y}
    \end{array}\right]=
    \overbrace{\left[\begin{array}{cc}
         P^{-\frac{1}{\alpha}} & 0 \\
         0 & P^{-\frac{1}{\alpha}} 
    \end{array}\right]}^{A}
     \left[\begin{array}{c}
         {x}  \\
         {y}
    \end{array}\right].
\end{equation}
Consequently, profiting from mapping theorem, we can conclude $\Phi \rightarrow \tilde{\Phi}$ where 
\begin{equation}
    \tilde{\lambda}=\text{det}[A^{-1}]\lambda=P^{\frac{2}{\alpha}}\lambda.
\end{equation}
\label{AP1}

\end{appendices}

\bibliographystyle{ieeetran.bst}
\bibliography{ref}

\begin{thebibliography}{10}

\bibitem{Ros}
{GSM Association (GSMA)}, ``{White paper: Study on socio-economic benefits of
  5G services provided in mmWave bands},'' tech. rep., ROSCONGRESS,
  December-2018.

\bibitem{Eric1-1}
{Ericsson Inc.}, ``{White paper: Ericsson Mobility Report June 2019},'' tech.
  rep., Ericsson, June 2019.

\bibitem{RAP1}
S.~{Rangan}, T.~S. {Rappaport}, and E.~{Erkip}, ``{Millimeter-Wave Cellular
  Wireless Networks: Potentials and Challenges},'' {\em Proceedings of the
  IEEE}, vol.~102, no.~3, pp.~366--385, 2014.

\bibitem{AccessSurvey}
J.~{Ding}, M.~{Nemati}, C.~{Ranaweera}, and J.~{Choi}, ``{IoT Connectivity
  Technologies and Applications: A Survey},'' {\em IEEE Access}, vol.~8,
  pp.~67646--67673, 2020.

\bibitem{OUT1}
K.~Ntontin and C.~Verikoukis, ``{Relay-Aided Outdoor-to-Indoor Communication in
  Millimeter-Wave Cellular Networks},'' {\em IEEE Systems Journal}, 2019.

\bibitem{Cov&rate}
T.~{Bai} and R.~W. {Heath}, ``{Coverage and Rate Analysis for Millimeter-Wave
  Cellular Networks},'' {\em IEEE Transactions on Wireless Communications},
  vol.~14, no.~2, pp.~1100--1114, 2015.

\bibitem{RAP2}
T.~S. Rappaport, R.~W. Heath~Jr, R.~C. Daniels, and J.~N. Murdock, {\em
  Millimeter wave wireless communications}.
\newblock Pearson Education, 2015.

\bibitem{RAP3}
T.~S. {Rappaport}, Y.~{Xing}, O.~{Kanhere}, S.~{Ju}, A.~{Madanayake},
  S.~{Mandal}, A.~{Alkhateeb}, and G.~C. {Trichopoulos}, ``{Wireless
  Communications and Applications Above 100 GHz: Opportunities and Challenges
  for 6G and Beyond},'' {\em IEEE Access}, vol.~7, pp.~78729--78757, 2019.

\bibitem{YangRel1}
G.~{Yang} and M.~{Xiao}, ``{Performance Analysis of Millimeter-Wave Relaying:
  Impacts of Beamwidth and Self-Interference},'' {\em IEEE Transactions on
  Communications}, vol.~66, no.~2, pp.~589--600, 2018.

\bibitem{KAVANPHD}
K.~Belbase, {\em {Analysis of Millimeter Wave Wireless Relay Networks}}.
\newblock PhD thesis, 2019.

\bibitem{RISNTO}
K.~RISNTO, M.~Renzo, J.~Song, F.~Lazarakis, J.~Rosny, D.-T. Phan-Huy,
  O.~Simeone, R.~Zhang, M.~Debbah, G.~Lerosey, M.~Fink, S.~Tretyakov, and
  S.~Shamai, ``{Reconfigurable Intelligent Surfaces vs. Relaying: Differences,
  Similarities, and Performance Comparison},'' 08 2019.

\bibitem{RISBasar1}
E.~{Basar}, M.~{Di Renzo}, J.~{De Rosny}, M.~{Debbah}, M.~{Alouini}, and
  R.~{Zhang}, ``{Wireless Communications Through Reconfigurable Intelligent
  Surfaces},'' {\em IEEE Access}, vol.~7, pp.~116753--116773, 2019.

\bibitem{RISRui1}
Q.~{Wu} and R.~{Zhang}, ``{Towards Smart and Reconfigurable Environment:
  Intelligent Reflecting Surface Aided Wireless Network},'' {\em IEEE
  Communications Magazine}, vol.~58, no.~1, pp.~106--112, 2020.

\bibitem{park2020extreme}
J.~Park, S.~Samarakoon, H.~Shiri, M.~K. Abdel-Aziz, T.~Nishio, A.~Elgabli, and
  M.~Bennis, ``{Extreme URLLC: Vision, Challenges, and Key Enablers},'' {\em
  arXiv preprint arXiv:2001.09683}, 2020.

\bibitem{wcnc}
M.~{Nemati}, J.~{Ding}, and J.~{Choi}, ``Short-range ambient backscatter
  communication using reconfigurable intelligent surfaces,'' in {\em 2020 IEEE
  Wireless Communications and Networking Conference (WCNC)}, pp.~1--6, 2020.

\bibitem{ref2}
J.~G. {Andrews}, T.~{Bai}, M.~N. {Kulkarni}, A.~{Alkhateeb}, A.~K. {Gupta}, and
  R.~W. {Heath}, ``{Modeling and Analyzing Millimeter Wave Cellular Systems},''
  {\em IEEE Transactions on Communications}, vol.~65, no.~1, pp.~403--430,
  2017.

\bibitem{ref3}
S.~{Singh}, F.~{Ziliotto}, U.~{Madhow}, E.~{Belding}, and M.~{Rodwell},
  ``{Blockage and directivity in 60 GHz wireless personal area networks: from
  cross-layer model to multihop MAC design},'' {\em IEEE Journal on Selected
  Areas in Communications}, vol.~27, no.~8, pp.~1400--1413, 2009.

\bibitem{ref4}
S.~{Biswas}, S.~{Vuppala}, J.~{Xue}, and T.~{Ratnarajah}, ``{On the Performance
  of Relay Aided Millimeter Wave Networks},'' {\em IEEE Journal of Selected
  Topics in Signal Processing}, vol.~10, no.~3, pp.~576--588, 2016.

\bibitem{GongSurvey}
S.~{Gong}, X.~{Lu}, D.~T. {Hoang}, D.~{Niyato}, L.~{Shu}, D.~I. {Kim}, and
  Y.~{Liang}, ``Towards smart wireless communications via intelligent
  reflecting surfaces: A contemporary survey,'' {\em IEEE Communications
  Surveys Tutorials}, pp.~1--1, 2020.

\bibitem{GenFadRIS}
I.~Trigui, W.~Ajib, and W.-P. Zhu, ``{A Comprehensive Study of Reconfigurable
  Intelligent Surfaces in Generalized Fading},'' {\em arXiv preprint
  arXiv:2004.02922}, 2020.

\bibitem{liu2020}
Y.~Liu, X.~Liu, X.~Mu, T.~Hou, J.~Xu, Z.~Qin, M.~Di~Renzo, and N.~Al-Dhahir,
  ``{Reconfigurable Intelligent Surfaces: Principles and Opportunities},'' {\em
  arXiv preprint arXiv:2007.03435}, 2020.

\bibitem{wu2020}
Q.~Wu, S.~Zhang, B.~Zheng, C.~You, and R.~Zhang, ``Intelligent reflecting
  surface aided wireless communications: A tutorial,'' {\em arXiv preprint
  arXiv:2007.02759}, 2020.

\bibitem{noma1}
T.~{Hou}, Y.~{Liu}, Z.~{Song}, X.~{Sun}, Y.~{Chen}, and L.~{Hanzo},
  ``Reconfigurable intelligent surface aided noma networks,'' {\em IEEE Journal
  on Selected Areas in Communications}, pp.~1--1, 2020.

\bibitem{noma2}
X.~Yue and Y.~Liu, ``{Performance analysis of intelligent reflecting surface
  assisted NOMA networks},'' {\em arXiv preprint arXiv:2002.09907}, 2020.

\bibitem{noma3}
C.~Zhang, W.~Yi, Y.~Liu, Z.~Qin, and K.~K. Chai, ``Downlink analysis for
  reconfigurable intelligent surfaces aided noma networks,'' {\em arXiv
  preprint arXiv:2006.13260}, 2020.

\bibitem{R1}
M.~A. Kishk and M.-S. Alouini, ``Exploiting randomly-located blockages for
  large-scale deployment of intelligent surfaces,'' {\em arXiv preprint
  arXiv:2001.10766}, 2020.

\bibitem{Mimo-1}
T.~Hou, Y.~Liu, Z.~Song, X.~Sun, Y.~Chen, and L.~Hanzo, ``{MIMO assisted
  networks relying on large intelligent surfaces: A stochastic geometry
  model},'' {\em arXiv preprint arXiv:1910.00959}, 2019.

\bibitem{R2}
M.~Di~Renzo and J.~Song, ``Reflection probability in wireless networks with
  metasurface-coated environmental objects: an approach based on random spatial
  processes,'' {\em EURASIP Journal on Wireless Communications and Networking},
  vol.~2019, no.~1, p.~99, 2019.

\bibitem{R3}
J.~He, K.~Yu, and Y.~Shi, ``Coordinated passive beamforming for distributed
  intelligent reflecting surfaces network,'' {\em arXiv preprint
  arXiv:2002.05915}, 2020.

\bibitem{AsymSINR}
Q.~{Nadeem}, A.~{Kammoun}, A.~{Chaaban}, M.~{Debbah}, and M.~{Alouini},
  ``{Asymptotic Max-Min SINR Analysis of Reconfigurable Intelligent Surface
  Assisted MISO Systems},'' {\em IEEE Transactions on Wireless Communications},
  pp.~1--1, 2020.

\bibitem{REF}
J.~G. Andrews, F.~Baccelli, and R.~K. Ganti, ``A tractable approach to coverage
  and rate in cellular networks,'' {\em IEEE Transactions on communications},
  vol.~59, no.~11, pp.~3122--3134, 2011.

\bibitem{ParkSG}
M.~{Rebato}, J.~{Park}, P.~{Popovski}, E.~{De Carvalho}, and M.~{Zorzi},
  ``{Stochastic Geometric Coverage Analysis in mmWave Cellular Networks With
  Realistic Channel and Antenna Radiation Models},'' {\em IEEE Transactions on
  Communications}, vol.~67, no.~5, pp.~3736--3752, 2019.

\bibitem{PPP}
D.~Moltchanov, ``{Distance distributions in random networks},'' {\em Ad Hoc
  Networks}, vol.~10, no.~6, pp.~1146 -- 1166, 2012.

\bibitem{parkSG1}
J.~{Park}, S.~{Kim}, and J.~{Zander}, ``{Tractable Resource Management With
  Uplink Decoupled Millimeter-Wave Overlay in Ultra-Dense Cellular Networks},''
  {\em IEEE Transactions on Wireless Communications}, vol.~15, no.~6,
  pp.~4362--4379, 2016.

\bibitem{ChanelSG2}
Y.~{Li}, J.~G. {Andrews}, F.~{Baccelli}, T.~D. {Novlan}, and C.~J. {Zhang},
  ``{Design and Analysis of Initial Access in Millimeter Wave Cellular
  Networks},'' {\em IEEE Transactions on Wireless Communications}, vol.~16,
  no.~10, pp.~6409--6425, 2017.

\bibitem{parkSG2}
J.~{Kim}, J.~{Park}, S.~{Kim}, S.~{Kim}, K.~W. {Sung}, and K.~S. {Kim},
  ``{Millimeter-Wave Interference Avoidance via Building-Aware Associations},''
  {\em IEEE Access}, vol.~6, pp.~10618--10634, 2018.

\bibitem{limani2020mmwave}
Z.~Limani, F.~Malandrino, C.~F. Chiasserini, and A.~Nordio, ``Mmwave beam
  management in urban vehicular networks,'' {\em arXiv preprint
  arXiv:2005.09726}, 2020.

\bibitem{parkBW}
K.~{Venugopal}, M.~C. {Valenti}, and R.~W. {Heath}, ``{Interference in
  finite-sized highly dense millimeter wave networks},'' in {\em 2015
  Information Theory and Applications Workshop (ITA)}, pp.~175--180, 2015.

\bibitem{kavan}
K.~{Belbase}, Z.~{Zhang}, H.~{Jiang}, and C.~{Tellambura}, ``{Coverage Analysis
  of Millimeter Wave Decode-and-Forward Networks With Best Relay Selection},''
  {\em IEEE Access}, vol.~6, pp.~22670--22683, 2018.

\bibitem{puo}
M.~{Nemati}, T.~{Baykas}, and J.~{Choi}, ``{Performance of TDOA and AOA
  Localization Techniques for Different Base-Stations Topologies},'' in {\em
  2019 13th International Conference on Signal Processing and Communication
  Systems (ICSPCS)}, pp.~1--7, 2019.

\bibitem{RuiRISNet}
J.~Lyu and R.~Zhang, ``{Hybrid Active/Passive Wireless Network Aided by
  Intelligent Reflecting Surface: System Modeling and Performance Analysis},''
  {\em arXiv preprint arXiv:2004.13318}, 2020.

\bibitem{BEAMBOOK}
V.~Giurgiutiu, ``{Chapter 6 - Piezoelectric Wafer Active Sensors},'' in {\em
  {Structural Health Monitoring of Aerospace Composites}} (V.~Giurgiutiu, ed.),
  pp.~177 -- 248, Oxford: Academic Press, 2016.

\bibitem{PLPL}
X.~{Zhang} and M.~{Haenggi}, ``{The Performance of Successive Interference
  Cancellation in Random Wireless Networks},'' {\em IEEE Transactions on
  Information Theory}, vol.~60, no.~10, pp.~6368--6388, 2014.

\bibitem{Champ}
M.~Haenggi and R.~Ganti, ``{Interference in Large Wireless Networks},'' {\em
  Foundations and Trends in Networking}, vol.~3, pp.~127--248, 01 2009.

\end{thebibliography}

\end{document}